\documentclass[journal,10pt]{IEEEtran}
\usepackage{graphicx}
\usepackage{mathrsfs}
\usepackage{graphicx}
\usepackage{amssymb}
\usepackage{indentfirst}
\usepackage{threeparttable}
\usepackage{booktabs}
\usepackage{subfigure}
\usepackage{graphicx}
\usepackage{algorithm,algorithmic}
\usepackage{multirow}
\usepackage{booktabs}
\usepackage{bm}
\usepackage{enumerate}
\usepackage{cite}
\usepackage{amsmath} 
\usepackage{amssymb}  

\usepackage{amsmath,amssymb,amsfonts}
\usepackage{graphicx}
\usepackage{subfigure}
\usepackage{bm}
\usepackage{color}

\newtheorem{theo}{Theorem}
\newtheorem{defi}{Definition}

\newtheorem{lemm}{Lemma}

\begin{document}

\title{Block Distributed Compressive Sensing Based  Doubly Selective  Channel Estimation and Pilot Design for  Large-Scale  MIMO  Systems}

\author{\IEEEauthorblockN{ Bo~Gong,  Lin~Gui, Qibo~Qin, Xiang~Ren, and Wen~Chen}

\IEEEauthorblockA{Department of Electronic Engineering, Shanghai Jiao Tong University, China\\
Email: \{gongbo;~ guilin;~ qinqibo;~ renx;~ wenchen\}@sjtu.edu.cn}

\thanks{Copyright (c) 2015 IEEE. Personal use of this material is permitted. However, permission to use this material for any other purposes must be obtained from the IEEE by sending a request to pubs-permissions@ieee.org.

B. Gong,  L. Gui, Q. Qin, and X. Ren are with the Department of Electronic
Engineering, Shanghai Jiao Tong University, Shanghai 200240, China (e-mail:
gongbo@sjtu.edu.cn). W. Chen is with Department of Electronic Engineering and Shanghai Key Laboratory of Navigation and Location Based Services lab, Shanghai Jiao Tong University.

This work was supported in part by the National Natural Science Foundation of China (61471236, 61420106008, 61671295, 61671294), the 111 Project (B07022), and the Shanghai Key Laboratory of Digital Media Processing; it is also partly sponsored by Shanghai Pujiang Program (16PJD029).
}}

\maketitle

\begin{abstract}
The doubly selective (DS) channel estimation in the large-scale multiple-input multiple-output (MIMO) systems is a challenging problem due to the large number of the channel coefficients to be estimated, which requires unaffordable and prohibitive pilot overhead. In this paper, firstly we conduct the analysis about the common sparsity of the basis expansion model (BEM) coefficients among all the BEM orders and all the transmit-receive antenna pairs. Then a novel pilot pattern is proposed,  which inserts the guard pilots to deal with the inter carrier interference (ICI) under the superimposed pilot pattern. Moreover, by exploiting the common sparsity of the BEM coefficients among different BEM orders and  different antennas, we propose a block  distributed compressive sensing (BDCS) based DS channel estimator for the large-scale MIMO systems. Its structured sparsity leads to the reduction of the pilot overhead under the premise of  guaranteeing the accuracy of the estimation. Furthermore, taking consideration of the block structure, a pilot design algorithm referred to as block discrete stochastic optimization (BDSO) is proposed. It optimizes the pilot positions by reducing the coherence among different blocks of the measurement matrix. Besides, a linear smoothing method is extended to large-scale MIMO systems to improve the accuracy of the estimation. Simulation results verify the performance gains of our proposed estimator and the pilot design algorithm compared with the existing schemes.

\end{abstract}

\begin{IEEEkeywords}
Block distributed compressive sensing,  doubly selective, large-scale MIMO, channel estimation, pilot design.
\end{IEEEkeywords}

\IEEEpeerreviewmaketitle

\section{Introduction}

Large-scale multiple-input multiple-output (MIMO) \cite{Bj??rnson2016}, \cite{Lu2014}  attracts much academic interest and  is considered as a promising technology in the incoming 5G cellular systems \cite{Larsson2014}. It enhances the data throughput and improves the link reliability of wireless communication systems by taking advantage of the high spatial multiplexing gains. In order to benefit from large-scale MIMO, one must obtain accurate channel state information (CSI)  which guarantees data recovery and contributes to multi-antenna array gains.

Time and frequency selective channel, which is also referred to as  doubly-selective (DS) channel,  is related to many wireless access links, such as  high-speed trains \cite{Ren2015} and millimeter-wave  communications \cite{Rangan2014}, \cite{Swindlehurst2014}. Frequency selectivity is caused by multipath propagation and time selectivity results from Doppler shift. For the DS channel estimation, inter-carrier interference (ICI) is a challenging problem, which incurs a large amount of channel coefficients to be estimated and the high complexity of the estimation schemes. In \cite{Ma2003, Tang2007,  Tang2011}, basis expansion model (BEM) was proposed to simplify the estimation process. The work  \cite{Ma2003} presented the research about the optimal training for DS channel estimation.  The work \cite{Tang2007} verified  several channel estimation schemes including the least squares (LS) estimator, the linear minimum mean square error (LMMSE) estimator and the best linear unbiased estimator (BLUE) combined with different BEM basis in single-input single-output (SISO) systems. And then in \cite{Tang2011} they were extended to MIMO systems with 2 transmit antennas.  However, in large-scale MIMO systems with more than a dozen or even dozens of antennas,  the number of channel coefficients to be estimated increases largely. The conventional schemes are not feasible since it requires  unaffordable pilot overhead and prohibitive complexity to avoid the inter-antenna interference (IAI)  and process the ICI. To our best knowledge, the  DS channel estimation in large-scale  MIMO systems has been seldom considered in the existing works.

Compressive sensing (CS) is an important framework to decrease the pilot overhead and the complexity of  channel estimation by taking advantage of the channel sparsity \cite{Eldar2012}.  The channel sparsity in three domains are concerned in the existing literature, including the delay-Doppler domain \cite{Ren2013}, the beam domain \cite{Gao2015}, and the delay domain \cite{Cheng2013}. In delay-Doppler domain,  a large Doppler shift incurs the leakage effects, which increases the channel sparsity and deteriorates the performance. Processing the leakage effects and enhancing the sparsity require high complexity of computation \cite{Taubock2010a}. The sparsity in beam domain relies on an open and wide propagation environment with few scatters. It is usually increased in the environment with rich reflection and the compressive sensing framework is no longer applicable. The most commonly exploited sparsity presents in delay domain in the existing CS based channel estimation schemes. It is analyzed that the broadband channels present sparsity in delay domain.

Various CS based channel estimation schemes appear recently. In \cite{Masood2015}, the authors proposed compressive estimation schemes for flat fading  channels in large-scale MIMO systems. It analyzed the channel sparsity among all the antennas and jointly estimated the channels of all the receive antennas in the base station. The works \cite{Rao2014, Qi2014, Gao2016, Nan2015, Hou2014} proposed  compressive CSI estimation schemes under frequency selective channels. In \cite{Rao2014}, the authors considered the channel model with the common  support and the individual support. Then they proposed an algorithm which adapted to this kind of specific structure for large-scale MIMO systems. The work \cite{Qi2014} proposed a scheme for the uplink channels in large-scale MIMO systems, which was relatively simple since it wasn't concerned with IAI. In \cite{Gao2016}, an approach  with an unknown sparsity was proposed, which was more realistic in the practical systems. The works \cite{Nan2015},\cite{Hou2014} suggested different block-structured pursuit algorithms. In \cite{Ren2015}, the authors  presented the research about DS channel estimation based on CS in SISO systems. It proposed a channel estimation scheme based on the position information in high mobility systems and optimized the corresponding pilot design. The work is extended to MIMO systems with 4 transmit antennas in \cite{Ren2013}, which is referred to as a low coherence compressed (LCC) channel estimation scheme. However, on the one hand, this scheme utilized the sparsity in delay-Doppler domain which is notably increased  by  a large Doppler shift,  and on the other hand, it could not support much more antennas. The work \cite{Cheng2013} introduced distributed compressive sensing (DCS) to DS channel estimation in SISO systems. It utilized the channel sparsity in delay domain and formulated the estimation into a DCS framework, which guaranteed a more accurate recovery.

Mutual coherence is an important factor concerned with the accuracy of the recovery in CS framework \cite{Eldar2012}. In SISO systems,  several algorithms were proposed to decrease the mutual coherence of the measurement matrix. In \cite{Cheng2013}, it proposed discrete stochastic optimization (DSO) to select the suboptimal pilot positions and \cite{Ren2013} proposed an algorithm to jointly optimize the pilot values and the pilot positions. In MIMO systems, \cite{Qi2015} and \cite{He2015} employed the optimization algorithms based on the stochastic search and the genetic algorithm. However, they were designed for orthogonal pilots, which meant the requirement of a large pilot overhead in large-scale MIMO systems. The works \cite{Gao2015}  adopted  the nonorthogonal pilots with equispaced pilot positions for the consideration of reducing the correlation among different virtual channels.

In this paper, we propose a block distributed compressive sensing (BDCS) based DS channel estimation scheme for the large-scale MIMO  systems and   a novel pilot design algorithm corresponding to the unique structure. They reduce the pilot overhead under the premise of guaranteeing the estimation accuracy. In specific, firstly we analyze the common sparsity of the BEM coefficients among all the BEM orders and all the transmit-receive antenna pairs in delay domain. Different from the sparsity analysis in the existing literature \cite{Gao2016, Nan2015, Hou2014}, we focus on the sparsity of the BEM coefficients rather than the channel coefficients.  In addition to the analysis of the common sparsity among different BEM orders for SISO systems in \cite{Cheng2013}, the common sparsity among different transmit antennas is considered here as well; Then,  a novel pilot pattern is proposed for DS channels in large-scale MIMO systems. It combines the property of the superimposed structure and the setting of the guard pilots. The superimposed structure reflects  the superimposed  pilot positions of different transmit antennas, which reduces the pilot overhead. The guard pilots are designed for the ICI of the DS channels, which avoid the contamination from the data subcarriers; Moreover, we propose a BDCS based channel estimator for DS channels in large-scale MIMO systems. The unknown BEM coefficients present block and common sparsity simultaneously.  The structure benefits the localization of the nonzero elements, and then leads to the performance improvement and the pilot overhead reduction; Furthermore, taking advantage of the analyzed block sparsity,  a novel  pilot design algorithm, referred to as block discrete stochastic optimization (BDSO), is proposed.  It also contributes to the performance gain and the spectral efficiency.

The remainder of this paper is organized as follows. Section II introduces the system model and the fundamentals of CS.  In section III,  we present the proposed channel estimation scheme and the proposed pilot design algorithm. In Section IV, we conduct the analysis of the complexity.   In section V, simulation results verify the validity of our work.  Section VI concludes this paper.

\emph{Notations}: $\left(  \cdot  \right)^T$ denotes matrix transpose, $\left(  \cdot  \right)^H$ represents matrix conjugate transpose.  $diag(\cdot)$ means a diagonal matrix, $\left|  \cdot  \right|$ denotes the absolute value, $\left\langle { \cdot , \cdot } \right\rangle $ denotes the inner product, ${\left\|  \cdot  \right\|_2}$ stands for the Euclidean norm, ${\left\|  \cdot  \right\|_0}$ denotes the number of nonzero values.
$\otimes$ represents Kronecker product. $\mathcal{S}$ indicates a set, ${A[m,n]}$ represents  the $(m+1,n+1)$-th  element of matrix $\mathbf A$. ${{[\mathbf A]}_\mathcal{S}}$ represents the selected rows of $\mathbf A$, whose indices correspond to the set $\mathcal{S}$. $\mathcal{CN}(0,{\sigma ^2})$  represents the complex Gaussian distribution with zero mean and $\sigma ^2$ variance. $\mathbf{I}_x$ means the identity matrix of order $x$. $vec(\mathbf{A})$ denotes the column-ordered vectorization of matrix $\mathbf{A}$.

\section{System Model and Fundamentals}

In this section, we introduce our system model, which includes the  transmission model and the channel model of the  large-scale MIMO-orthogonal frequency division multiplexing (OFDM) systems. Besides, the fundamental knowledge of CS and structured compressive sensing (SCS) is briefly illustrated.
\subsection{System Model}
\subsubsection{Transmission Model}
\begin{figure}[t]
\centering
\includegraphics[scale=0.5]{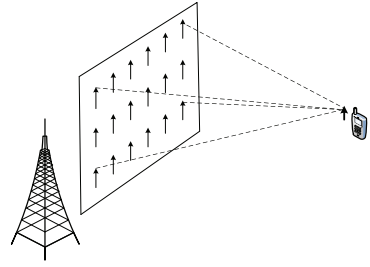}
\caption{Transmission model}
\label{1}
\end{figure}

We consider a large-scale MIMO-OFDM system. The base station is equipped  with a great many  antennas. It serves a number of terminals with a single antenna. As shown in Fig. 1, the antenna array is arranged in a rectangle, which consists of ${N_B}$ antennas.
We adopt  frequency division duplex (FDD) mode in our system. In the OFDM system, for any transmit-receive antenna pair, there exist $N$ subcarriers in a parallel transmission. A part  of the subcarriers  are selected as pilot subcarriers to estimate the channel coefficients and the remaining ones are responsible for data.
\subsubsection{Channel Model}

We consider a DS channel model. The multipath effect leads to the frequency selectivity and the Doppler shift results in the time selectivity.  For each transmit-receive antenna pair between the user side and the base station, the channel in  time domain can be assumed to be a finite impulse response (FIR) filter. Let $h^{({n_B})}[n,l]$ represent the channel coefficient  of the $(l+1)$-th tap at the $(n+1)$-th  instant of the channel between the ${n_B}$-th antenna in the base station and the terminal, in which ${n_B} \in [1,N_B]$, $l \in  [0, L-1]$, and $n \in [0,N-1]$. As $L$ is the length of the channel, we have
\begin{equation}
h^{({n_B})}[n,l]=0, \quad  l < 0 \quad or \quad  l  \ge  L.
\end{equation}
Assume that $\mathbf{H}_t^{({n_B})} \in {\mathbb{C}^{N \times N}}$ describes channel matrix in time domain, and its elements can be expressed as
\begin{equation}\label{d1}
{{H^{({n_B})}_t} [p,q]} = h^{({n_B})}[p,\bmod (p - q,N)], \quad  p,q \in [0,N - 1].
\end{equation}
 The channel matrix in frequency domain ${\mathbf{H}_f^{({{n_B}})}}$ can be derived from
 \begin{equation}\label{d2}
  \mathbf{H}_f^{({n_B})}=\mathbf{W}{\mathbf{H}_t^{({{n_B}})}}\mathbf{W}^H,
\end{equation}
in which $\mathbf{W}$ is the discrete fourier transform (DFT) matrix and  ${W[m,n]} = {N^{-1/2}}\exp \left( {-j2\pi mn/N} \right)$, $m,n \in [ {0,N-1} ]$, $j^2=-1$.

It is found  that for DS channels, we have to estimate the channel coefficients of each channel tap at each time instant. The total number of coefficients to be estimated for a transmit-receive antenna pair is $NL$.  $N$ is the number of the samplings in time domain, which is equal to the number of the subcarriers, and  $L$ is the number of multipaths, i.e.,  the length of the channel. BEM \cite{Tang2011} is an important technique for DS channel estimation, which is always introduced to reduce the number of the  coefficients  to be estimated.  Let ${\mathbf{h}}_l^{(n_B)} {=} {( {
{h^{(n_B)}[0,l]}, \cdots ,{h^{(n_B)}[N - 1,l]}  } )^T}  \in {\mathbb{C}^{N \times 1}}$ denote the channel coefficients of the $l$-th channel tap and the $n_B$-th transmit antenna. Each ${\mathbf{h}^{(n_B)}_l}, l \in [0,L-1], n_B \in [1, N_B]$ can be expressed as
\begin{equation}\label{d3}
{{\mathbf{h}}^{(n_B)}_l} = {\mathbf{V}}{{\bm{\theta}}^{(n_B)}_l} + {{\bm{\varepsilon }}^{(n_B)}_l},
\end{equation}
 in which,
 \[{{\bm{\theta}}^{(n_B)}_l} = {
({\theta^{(n_B)}[0,l]},  {\theta^{(n_B)}[1,l]},  { \cdots },  {\theta^{(n_B)}[D - 1,l]})^T
} \in  {\mathbb{C}^{D \times 1}}\] is the BEM coefficients, and \[{\bm{\varepsilon}^{(n_B)} _l} = {({\varepsilon^{(n_B)} [0,l]}, {\varepsilon^{(n_B)} [1,l]},  \cdots, {\varepsilon ^{(n_B)}[N - 1,l]})^T}\in  {\mathbb{C}^{N \times 1}}\] is the BEM modeling error. Besides, ${\mathbf{V}} {=} ({{{\mathbf{v}}_0}}, {{{\mathbf{v}}_1}}, \cdots, {{{\mathbf{v}}_{D - 1}}})$, in which, ${\mathbf{v}_d}$ is the BEM basis function, $d \in [0,D-1]$,  and $D$ $(D \ll N)$ is the number of the BEM orders. Apparently the number of channel coefficients to be estimated is reduced from $NL$ to $DL$ for one transmit-receive antenna pair. The vector of the channel taps for the $n_B$-th transmit antenna can be formulated as
\begin{equation}\label{recovery}
\overline {\mathbf{h}}^{(n_B)}  = \left( {{\mathbf{V}} \otimes {{\mathbf{I}}_L}} \right)\overline {\bm{\theta}}^{(n_B)}  + \overline {\bm{\varepsilon }}^{(n_B)} ,
\end{equation}
in which
\begin{equation}
\begin{split}
{\overline{\mathbf{h}}}^{(n_B)}  {=} {( {{{(\widetilde {\mathbf{h}}}^{(n_B)}_0)^T}, \cdots ,{{(\widetilde {\mathbf{h}}}^{(n_B)}_{N - 1})^T}})^T},  \\
{\overline {\bm{\theta}}}^{(n_B)}  {=} {( {{{(\widetilde {\bm{\theta}}}^{(n_B)}_0)^T}, \cdots ,{{(\widetilde {\bm{\theta}}}^{(n_B)}_{D {-} 1}}})^T)^T},   \\
 \overline {\bm{\varepsilon }}^{(n_B)}  {=} {( {(\widetilde {\bm{\varepsilon }}_0^{(n_B)})^T, \cdots ,(\widetilde {\bm{\varepsilon }}_{N - 1}^{(n_B)})^T})^T},
 \end{split}
\end{equation}
 with
 \begin{equation}  \nonumber
\begin{split}
{\widetilde {\mathbf{h}}^{(n_B)}_n} = {( {h^{(n_B)}[n,0],  \cdots,  h^{(n_B)}[n,L - 1])}^T} {\in} {\mathbb{C}^{L \times 1}},  \\
{\widetilde {\bm{\theta}}^{(n_B)}_d} = {({\theta^{(n_B)}[d,0],  \cdots,  \theta^{(n_B)}[d,L - 1])}^T} {\in} {\mathbb{C}^{L \times 1}}, \\
{\widetilde {\bm{\varepsilon }}^{(n_B)}_n} = {( {{\varepsilon }^{(n_B)}[n,0],  \cdots,  {\varepsilon }^{(n_B)}[n,L - 1]})^T} {\in} {\mathbb{C}^{L \times 1}},
\end{split}
\end{equation}
for $n \in [0,N-1]$ and $d \in [0,D-1]$. $\widetilde {\mathbf{h}}^{(n_B)}_n$ represents the channel coefficients at the $(n+1)$-th time instant of the $n_B$-th antenna. $\widetilde {\bm{\theta}}^{(n_B)}_d$ is the BEM coefficients of the $d$-th order and the $n_B$-th antenna. ${\widetilde {\bm{\varepsilon }}^{(n_B)}_n}$ is the modeling error.

Now we  briefly derive the expression of the BEM in frequency domain. From
(\ref{d1}) and (\ref{d3}), by simple arrangement and observation, we have
\begin{equation}\label{7}
{\mathbf{H}}_t^{({n_B})} = \sum\limits_{d = 0}^{D - 1} {diag{( {{\mathbf{v}}_d})} } {\tilde{\mathbf \Theta }}_d^{({n_B})} + {{{{\mathbf  E}}}^{({n_B})}},
\end{equation}
where ${\tilde{\mathbf \Theta }}_d^{({n_B})}$ is a circulant matrix with ${\widetilde {\bm{\theta}}^{(n_B)}_d} {=} {[{\theta^{(n_B)}[d,0],  \cdots,  \theta^{(n_B)}[d,L - 1]]}^T}$ as its first column \cite{Tang2011}. Due to its circularity, ${\tilde{\mathbf \Theta }}_d^{({n_B})}$ can be diagonalized as
\begin{equation}\label{8}
{\tilde{\mathbf \Theta }}_d^{({n_B})} = {{\mathbf{W}}^H}diag{( {{\mathbf{W}}_L}{\tilde{\bm \theta }}_d^{({n_B})})} {\mathbf{W}},
\end{equation}
where ${{\mathbf{W}}_L}$ denotes the submatrix that extracts the first $L$ columns of $\mathbf{W}$. Accordingly, substitute (\ref{8}) into (\ref{7}) and  the time domain channel matrix can be denoted as
\begin{equation}\label{td}
{\mathbf{H}}_t^{({n_B})} = \sum\limits_{d = 0}^{D - 1} {diag{( {{\mathbf{v}}_d})} {{\mathbf{W}}^H}diag{( {{\mathbf{W}}_L}{\tilde{\bm \theta }}_d^{({n_B})})} {\mathbf{W}}}  + {{\mathbf{{E}}}^{({n_B})}}.
\end{equation}
Substituting (\ref{td}) into (\ref{d2}), it is not hard to find that the channel matrix in frequency domain can be expressed as
\begin{equation}\label{matfre}
 {{\mathbf{H}}^{(n_B)}_f} = \sum\limits_{d = 0}^{D - 1} {{{\mathbf{V}}_d}{{\bm{\Theta}}^{(n_B)}_d}}  + {{\bm{\Delta }}^{(n_B)}},
\end{equation}
 in which,
 \begin{align}
  &{\mathbf{V}_d} = {\mathbf{W}}diag\left( {{\mathbf{v}_d}} \right)\mathbf{W}^H,  \notag\\ &{\bm{\Theta}^{(n_B)}_d} = diag{({{\sqrt N} {\mathbf{W}}{(
{{{\widetilde{\bm{\theta}}}{^{(n_B)}_d}}^T,} {{\mathbf{0}_{1 \times (N - L)}}
})}^T}),}    \notag
\end{align}
and ${{\bm{\Delta }}^{(n_B)}}$ is the modeling error \cite{Tang2007}.

\subsection{CS and SCS}

In this part, the basic knowledge of CS and SCS are introduced. The SCS means that the sparsity presents a certain structure, including DCS, block compressive sensing (BCS), and BDCS here.

\subsubsection{CS}
CS is an attractive framework, which recovers a high-dimensional sparse signal from a low dimensional observed vector.  It solves the  underdetermined problem
\begin{equation}
\mathbf{r} = \mathbf{A} \mathbf{x}+\mathbf{e},
\end{equation}
 in which $\mathbf{x} \in \mathbb{C}^{Z \times 1}$ is an unknown high-dimensional vector,  $\mathbf{A} \in \mathbb{C}^{M \times Z} (M \ll Z)$  is the measurement matrix, $\mathbf{r} \in \mathbb{C}^{M \times 1}$ represents the observed low-dimensional vector, and $\mathbf{e}$ denotes the noise term. The theory of CS is based on two important premises:
 \begin{itemize}
  \item The first one is the sparsity of the high-dimensional vector, which means that  $\mathbf{x}$ is a sparse vector with sparsity $K$, i.e., ${\left\| \mathbf{x} \right\|_{{0}}} = K, K \ll Z$.
   \item The second one is that the measurement matrix $\mathbf{A}$ satisfies restricted isometry  property (RIP) condition \cite{Eldar2012}.
 \end{itemize}
  If these two conditions are satisfied, a high probability of the exact recovery of $\mathbf{x}$ can be guaranteed. But it should be noted that it is difficult to  verify RIP condition due to the prohibitive complexity and the tremendous computation. In practical schemes, mutual coherence property (MCP) \cite{Duarte2011} is an important reference value of the measurement matrix, which reflects the coherence between columns.
\begin{defi}
The MCP of a matrix $\mathbf{A}$ is
\begin{equation}
\mu(\mathbf{A})=\max_{1\leq i\neq j\leq Z}\frac{\vert\left\langle \mathbf{a}_{i},\mathbf{a}_{j}\right\rangle \vert}{{\lVert\mathbf{a}_{i}\rVert}_{2}{\lVert\mathbf{a}_{j}\rVert}_{2}},
\end{equation}
\end{defi}
where ${\mathbf a}_{i}$ and ${\mathbf a}_{j}$ denote the $i$-th and the $j$-th columns of ${\mathbf A}$.
\begin{lemm}[\cite{Donoho2006}]
Suppose that $\mathbf{A}$ has MCP $\mu$ and that the sparsity of $\mathbf{x}$ is $K$ with $K < ({\raise0.7ex\hbox{$1$} \!\mathord{\left/
 {\vphantom {1 \mu }}\right.\kern-\nulldelimiterspace}
\!\lower0.7ex\hbox{$\mu $}} + 1)/4$. Furthermore, suppose that we obtain measurements of the form ${\mathbf{r}} = {\mathbf{Ax}} + {{\mathbf{e}}}$. Then when the set of solutions ${\mathop{\rm B}\nolimits} (\mathbf{r}) = \{ \mathbf{z}:{{\left\| {\mathbf{Az} - \mathbf{r}} \right\|}_2} \le \rm{\zeta} \} $, $\rm{\zeta}$ is a constant value,  the solution ${\hat {\mathbf x}}$ obeys
\begin{equation}
{\left\| {{\mathbf{x}} - {\hat {\mathbf x}}} \right\|_2} \le \frac{{{{\left\| {{\mathbf{e}}} \right\|}_2} + {\rm{\zeta }}}}{{\sqrt {1 - \mu (4K - 1)} }}.
\end{equation}
\end{lemm}

Lemma 1 verifies that a smaller value of the MCP will lead to a more accurate recovery of $\mathbf{x}$.  Basis pursuit (BP) \cite{Chen1998} and orthogonal matching pursuit (OMP) \cite{Mallat1993} are the widely adopted recovery algorithms of CS.

\subsubsection{SCS}
 Several models of SCS are elaborated as follows. We summarize the properties of the problems and the corresponding recovery algorithms.
\begin{itemize}
\item \emph{DCS}:
 Many applications are concerned with a problem that  several high-dimensional sparse vectors with the same positions of  nonzero elements are compressed by a common measurement matrix,
\begin{equation}\label{cse}
\mathbf{R}_j = \mathbf{A}{\mathbf{X}_j} + {{\mathbf{w}}_j}, \quad  j \in [0,J-1].
\end{equation}
It is inefficient to recover each sparse vector separately. DCS framework is applied to jointly compress and recover the multiple correlated sparse  signals.  The basic form of DCS is
\begin{equation}\label{dcse}
\mathbf{R} = \mathbf{A}\mathbf{X} + \mathbf{w},
\end{equation}
 in which,
 \begin{align}
 \mathbf{R} &= [{\mathbf{R}_0}, \cdots ,{\mathbf{R}_{J-1}}] \in {\mathbb{C}^{M \times J}}, \notag \\ \mathbf{X} &= [{\mathbf{X}_0}, \cdots , {\mathbf{X}_{J-1}}] \in {\mathbb{C}^{Z \times J}},  \notag
 \end{align}
 and
\begin{equation}
 {\mathbf{w}} = [{{\mathbf{w}}_0}, \cdots ,{{\mathbf{w}}_{J - 1}}] \in {\mathbb{C}^{M \times J}} \notag
 \end{equation}
  is the noise matrix.  All the columns of $\mathbf{X}$ share the same nonzero positions. The recovery accuracy also relies on the common sparsity of $\mathbf{X}$ and the property of the measurement matrix $\mathbf{A}$ as CS above. As proved in \cite{Baron2005a},  DCS provides higher accuracy of the recovery with fewer observed values than CS by utilizing the common sparsity. The reason is that multiple vectors contribute to the localization of the nonzero elements. Simultaneous-OMP (SOMP) \cite{Baron2005a} is an important algorithm for the recovery of DCS.
 \item \emph{BCS}: The BCS means that the sparsity of the unknown high-dimensional vector presents block sparsity. The basic form of the problem is
\begin{equation}
{\mathbf{r}}' = \mathbf{A}{\mathbf{x}}'+{\mathbf{e}}',
\end{equation}
in which, for better illustration of the block sparsity of ${\mathbf{x}}'$, it is decomposed as
\begin{equation}
{\mathbf{x}}' = {\left[
{{\mathbf{x}}_1^{T}},{{\mathbf{x}}_2^{T}}, \cdots,{{\mathbf{x}}_{T}^{T}} \right]}^{T}, \quad   {\mathbf{x}}_t \in {{\mathbb{C}}^{d \times 1}}, \, t \in [1,T],
\end{equation}
and accordingly the measurement matrix $\mathbf{A}$ can be decomposed as
\begin{equation}
{\mathbf{A}} = \left[
{{{\mathbf{A}}_1}}, {{{\mathbf{A}}_2}}, \cdots, {{{\mathbf{A}}_{{T}}}} \right], \quad \mathbf{A}_t \in  {\mathbb{C}^{M \times d}}.
\end{equation}
We can see that the unknown high-dimensional vector ${\mathbf{x}}'$ is decomposed into $T$ parts. The $K$ sparsity indicates that $K$ of them are nonzero blocks while the remaining $T-K$ parts are zero blocks, $K \ll T$. The nonzero block means that all the elements of the block are nonzero, and the zero block is constituted by $d$ zero elements. We set an example for further clarification about the block sparsity with assumption of $T=6$ and $K=2$ as
\begin{equation}
{\mathbf{x}}' = \left[ {\begin{array}{*{20}{c}}
{\underbrace {{\mathbf{x}}_1^{T}}_{\mathbf{0}}}&{\underbrace {{\mathbf{x}}_2^{T}}_{{\mathbf{x}}_2^{T}}}&{\underbrace {{\mathbf{x}}_3^{T}}_{\mathbf{0}}}&{\underbrace {{\mathbf{x}}_4^{T}}_{\mathbf{0}}}&{\underbrace {{\mathbf{x}}_5^{T}}_{{\mathbf{x}}_5^{T}}}&{\underbrace {{\mathbf{x}}_6^{T}}_{\mathbf{0}}}
\end{array}} \right]^{T},
\end{equation}
in which, ${{\mathbf{x}}_2}$, ${{\mathbf{x}}_5}$ are the nonzero blocks while ${{\mathbf{x}}_1}$, ${{\mathbf{x}}_3}$, ${{\mathbf{x}}_4}$, ${{\mathbf{x}}_6}$ are the zero blocks. The recovery algorithm, block orthogonal matching pursuit (BOMP), is proposed in \cite{He2015} to solve the problem of BCS.  The block sparsity is utilized to improve the accuracy of the nonzero elements localization and contributes to the performance gain.
\item \emph{BDCS}: The BDCS combines the property of the DCS and the BCS. The structure includes both the block sparsity of each unknown high-dimensional vector and the common sparsity among different unknown high-dimensional vectors. The form of the BDCS is described as
\begin{equation}
{{\mathbf{R}}'} = {\mathbf{A}}{{\mathbf{X}}'}.
\end{equation}
Similar to BCS, the measurement matrix $\mathbf{A}$ is decomposed as
\begin{equation}
{\mathbf{A}} = \left[
{{{\mathbf{A}}_1}}, {{{\mathbf{A}}_2}}, \cdots, {{{\mathbf{A}}_{{T}}}} \right], \quad \mathbf{A}_t \in  \mathbb{C}^{M \times d},
\end{equation}
and the unknown matrix ${\mathbf{X}}'$ can be decomposed into $T \times J$ parts as
\begin{equation}
{{\mathbf{X}}'} = \left[ {\begin{array}{*{20}{c}}
{{\mathbf{x}}_{11}}& \cdots &{{\mathbf{x}}_{1J}}\\
{{\mathbf{x}}_{21}}& \cdots &{{\mathbf{x}}_{2J}}\\
 \vdots & \ddots & \vdots \\
{{\mathbf{x}}_{T1}}& \cdots &{{\mathbf{x}}_{TJ}}
\end{array}} \right],
\end{equation}
in which, ${\mathbf{x}}_{tj} {\in} \mathbb{C}^{d \times 1}$,  $t {\in} [1,T], \, j {\in} [1,J]$.
The block sparsity is reflected on the $j$-th column of ${{\mathbf{X}}'}$, $j \in [1,J]$. $K$ out of the $T$ parts are nonzero blocks while the remaining ones are zero blocks. The common sparsity means that the position of the $K$ nonzero blocks are the same among all the $J$ columns of ${{\mathbf{X}}'}$. We further explain the block and common sparsity with an example under the assumption of $K=2, \, T=6, \, J=3$.
\begin{equation}
{{\mathbf{X}}'} = \left[ {\begin{array}{*{20}{c}}
{\underbrace {{\mathbf{x}}_{11}^{T}}_{\mathbf{0}}}&{\underbrace {{\mathbf{x}}_{21}^{T}}_{{\mathbf{x}}_{21}^{T}}}&{\underbrace {{\mathbf{x}}_{31}^{T}}_{\mathbf{0}}}&{\underbrace {{\mathbf{x}}_{41}^{T}}_{\mathbf{0}}}&{\underbrace {{\mathbf{x}}_{51}^{T}}_{{\mathbf{x}}_{51}^{T}}}&{\underbrace {{\mathbf{x}}_{61}^{T}}_{\mathbf{0}}}\\
{\underbrace {{\mathbf{x}}_{12}^{T}}_{\mathbf{0}}}&{\underbrace {{\mathbf{x}}_{22}^{T}}_{{\mathbf{x}}_{22}^{T}}}&{\underbrace {{\mathbf{x}}_{32}^{T}}_{\mathbf{0}}}&{\underbrace {{\mathbf{x}}_{42}^{T}}_{\mathbf{0}}}&{\underbrace {{\mathbf{x}}_{52}^{T}}_{{\mathbf{x}}_{52}^{T}}}&{\underbrace {{\mathbf{x}}_{62}^{T}}_{\mathbf{0}}}\\
{\underbrace {{\mathbf{x}}_{13}^{T}}_{\mathbf{0}}}&{\underbrace {{\mathbf{x}}_{23}^{T}}_{{\mathbf{x}}_{23}^{T}}}&{\underbrace {{\mathbf{x}}_{33}^{T}}_{\mathbf{0}}}&{\underbrace {{\mathbf{x}}_{43}^{T}}_{\mathbf{0}}}&{\underbrace {{\mathbf{x}}_{53}^{T}}_{{\mathbf{x}}_{53}^{T}}}&{\underbrace {{\mathbf{x}}_{63}^{T}}_{\mathbf{0}}}
\end{array}} \right]^T.
\end{equation}
The recovery algorithms, structured subspace pursuit (SSP) \cite{Gao2014} and block simultaneous orthogonal matching pursuit (BSOMP) \cite{Qin2016},  exploit both the common sparsity and the block sparsity  to get the more accurate recovery performance.
\end{itemize}

\section{The Proposed Estimator}

In this section, a novel DS channel estimator is proposed for large-scale MIMO systems by extending the DS channel estimation scheme  in SISO systems \cite{Cheng2013}. In the process of extension,  firstly we propose a novel pilot pattern which additionally considers the superimposed pilot structure of all the antennas. Then the estimator is formulated as a BDCS based problem which is different from the DCS based problem for SISO systems. Furthermore exploiting the block sparsity, the pilot design algorithm optimizes the coherence among different blocks of the measurement matrix in comparison to the different columns in SISO systems.

\subsection{The common sparsity of BEM coefficients}

In this part, we analyze the common sparsity of the BEM coefficients for large-scale MIMO systems. In addition to the common sparsity among different BEM orders, which is discussed in \cite{Cheng2013} for SISO systems, we consider the common sparsity among different antennas here as well.
\begin{theo}
 The elements of  the BEM coefficients set $\{ {\widetilde {\bm{\theta}}} _d^{(n_B)}\}$,  $n_B \in [1, N_B]$, $d \in [0,D-1]$, in a large-scale MIMO system share common sparsity among different BEM orders and different transmit antennas under the condition of $\frac{{{s_{\max }}}}{C} \le \frac{1}{{10BW}}$, in which $s_{max}$ denotes the maximum distance  between any two  transmit antennas, $C$ is the speed of light and $BW$ is the signal bandwidth.
\end{theo}
\begin{proof}
The proof is conducted by three steps. Firstly we analyze the common sparsity of the channel coefficients  in delay domain among all the sampling instants.  Then the common sparsity of the channel coefficients among different antennas is illustrated. Finally, exploiting the relationship between the channel coefficients and the BEM coefficients, we prove the common sparsity of the BEM coefficients among different BEM orders and different antennas. \begin{lemm}[\cite{Cheng2013}]\label{lemm1}
The channel coefficients of the $n_B$-th $n_B \in [1, N_B]$ transmit-receive antenna pair $\{\widetilde {\mathbf{h}}^{(n_B)}_n\}{\in} {\mathbb{C}^{L \times 1}}$   have common sparsity among all the sampling instants $n \in [0,N-1]$, i.e., their nonzero positions are the same.
\end{lemm}

Assume that there are $L$ channel taps  for a  transmit-receive antenna pair and the index set of them is denoted  as  $[0,L-1]$. Conventionally,  there exist  $K$ $(K \ll L)$ nonzero taps, $\left\{ {l_1}, \cdots ,{l_K} \right\}\subset [0,L-1]$, which means that
 \begin{equation}
{\mathbf{h}}_l^{({n_B})}= {({h^{({n_B})}}[0,l], \cdots ,{h^{({n_B})}}[N - 1,l])^T}
 = {\mathbf{0}},
 \end{equation}
 for any $l \notin \left\{ {l_1}, \cdots ,{l_K} \right\}$.
It is not hard to find that $\{\widetilde {\mathbf{h}}^{(n_B)}_n\}$  $n \in [0,N-1]$ are all $K$-sparse vectors
 \begin{equation}
 \begin{split}
 {\widetilde {\mathbf{h}}^{(n_B)}_n} &= {( {h^{(n_B)}[n,0],  \cdots,  h^{(n_B)}[n,L - 1])}^T}\\
 &=(0, \cdots, {h^{({n_B})}}[n,{l_1}], \cdots, 0, \cdots, {h^{({n_B})}}[n,{l_K}], \cdots, 0)^T
\end{split}
 \end{equation}
 and their common nonzero positions in delay domain are $\left\{ {l_1}, \cdots, {l_K} \right\}$.
\begin{lemm}[\cite{Masood2015}]\label{lemm2}
In large-scale MIMO systems, all the transmit-receive antenna pairs share common sparsity in delay domain  if $\frac{{{s_{\max }}}}{C} \le \frac{1}{{10BW}}$, in which $s_{max}$ denotes the maximum distance  between any two  transmit antennas, $C$ is the speed of light and $BW$ is the signal bandwidth.
\end{lemm}

As referred in \cite{Masood2015}, if $\frac{{{s_{\max }}}}{C} \le \frac{1}{{10BW}}$, the links of all the transmit-receive antenna pairs  scatter invariantly in space. Thus the indices of their strong channel taps are the same,  i.e., $\left\{ {l_1}, \cdots ,{l_K} \right\}  \subset [0,L - 1]$. In another word, they share common sparsity in delay domain.
It is safe to assume that all the transmit-receive antenna pairs have the same nonzero positions of the channel taps in a large-scale MIMO system, which means that the elements of $\{\widetilde {\mathbf{h}}^{(n_B)}_n\}$, $n_B \in [1,N_B]$ have common sparsity among different antennas. As to the condition of $\frac{{{s_{\max }}}}{C} \le \frac{1}{{10BW}}$,  in the long term evolution (LTE) systems \cite{Rangan2014} with parameters of BW=20MHz and the center frequency of 2.6GHz, we have that the 10$\times$10 transmit antenna array has common channel support \cite{Masood2015}.

Combined Lemma \ref{lemm1} and Lemma \ref{lemm2} together, we conclude that the elements of $\{\widetilde {\mathbf{h}}^{(n_B)}_n\}$, $n_B \in [1,N_B]$, $n \in [0,N-1]$ have common sparsity among all the sampling instants and all the antennas. Now we will discuss the relationship of the BEM coefficients and the channel coefficients. For the convenience of elaboration, we ignore the modeling error and get
\begin{equation}
 {\mathbf{h}^{(n_B)} _l} \approx {\mathbf{V}}{{\bm{\theta}}^{(n_B)} _l}.
\end{equation}
The linear relationship reflects that
 ${{\bm{\theta}}^{(n_B)} _l}{=}\mathbf{0}$, $l {\notin} \left\{ {l_1}, \cdots ,{l_K} \right\}$,
 since
 ${{\mathbf{h}}^{(n_B)} _l}=\mathbf{0}$, $l \notin \left\{ {l_1}, \cdots ,{l_K} \right\}$.
 Apparently, the elements of  $\{{\widetilde {\bm{\theta}}^{(n_B)}_d}\}$,  $n_B \in [1,N_B]$, $d \in [0,D-1]$ have common sparsity among different BEM orders and different antennas.
\end{proof}

\subsection{The Proposed Pilot Pattern}

In this part, a novel pilot pattern is proposed for DS channels in large-scale MIMO systems, which combines the guard pilot design and the superimposed structure among different antennas to combat the ICI and reduce the pilot overhead.

\begin{figure}[t]
\centering
\includegraphics[scale=0.5]{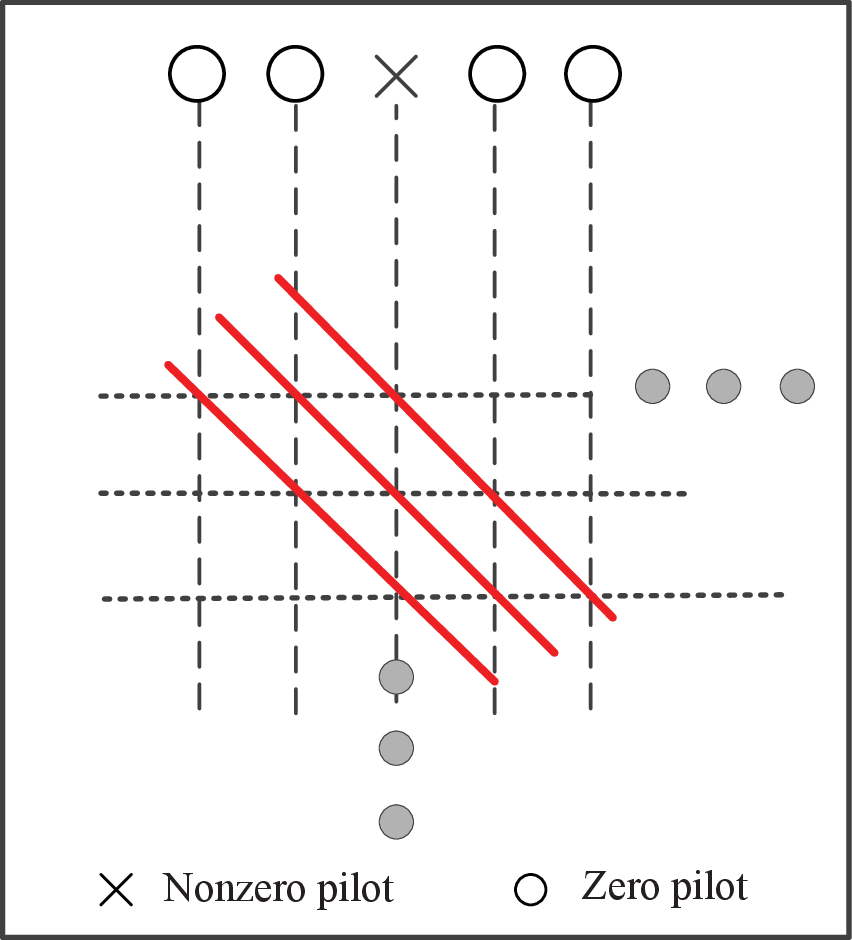}
\caption{ The diagram of CE-BEM.}
\label{0}
\end{figure}

Complex exponential basis expansion model (CE-BEM) is adopted due to its simple form. We have that the basis function  of the CE-BEM ${{\mathbf{v}}_d} = {(1, \cdots, {{e^{j\frac{{2\pi }}{N}n(d - \frac{{D - 1}}{2})}}}, \cdots, {{e^{j\frac{{2\pi }}{N}(N - 1)(d - \frac{{D - 1}}{2})}}})^T}$, $d \in [0, D-1]$, $j^2=-1$. In Fig. \ref{0}, we present the geometric expression of the CE-BEM for better illustration. A square, consisting of $N$ rows and $N$ columns, denotes the wireless channel between a transmit-receive antenna pair. The $N$ rows represent the $N$ transmitted subcarriers and the $N$ columns correspond to  the $N$ received subcarriers. The intersection of the $i$-th row and the $i$-th column means the channel coefficient of the $i$-th subcarrier, $i \in [1,N]$ while the  intersection of the $i$-th row and the $j$-th column means the interference to the $j$-th subcarrier from the $i$-th subcarrier, $j \in [1,N], j  \ne i$. Under the function of the CE-BEM with order $D$, the square is reduced to $D$ diagonals  as shown by the red lines in Fig. \ref{0}. To combat the ICI, which is represented by the $D-1$ subdiagonals, the nonzero pilots are always accompanied by several zero guard pilots on both sides. It is derived in \cite{Ma2003} that the optimal number of guard pilots on one side is $D-1$. Besides, in \cite{Tang2007}, we can see that only the central $D$ pilots including the nonzero one are not contaminated by the data subcarriers.
\begin{figure}[t]
\centering
\subfigure[A pilot pattern for DS channels in SISO systems.]{
\begin{minipage}[b]{0.5\textwidth}\label{A1}
\includegraphics[width=1\textwidth, height=1.5cm]{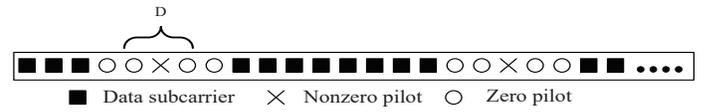}
\end{minipage}}

\subfigure[A superimposed pilot pattern for frequency selective channels in large-scale MIMO systems.]{
\begin{minipage}[b]{0.5\textwidth}\label{A2}
\includegraphics[width=1\textwidth,height=2cm]{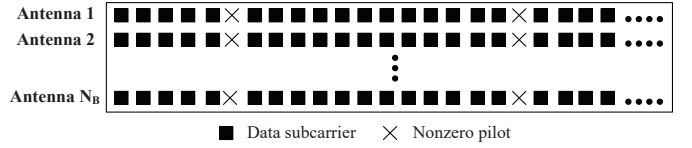}
\end{minipage}}

\subfigure[Our proposed pilot pattern for DS channels in large-scale MIMO systems.]{
\begin{minipage}[b]{0.5\textwidth}\label{B2}
\includegraphics[width=1\textwidth,height=2.5cm]{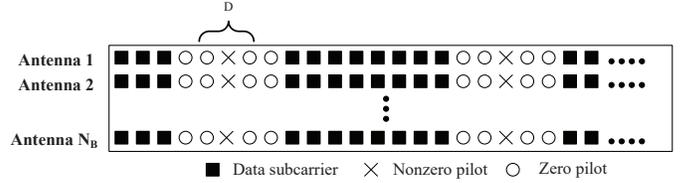}
\end{minipage}}

\subfigure[An orthogonal pilot pattern for DS channels in MIMO systems.]{
\begin{minipage}[b]{0.5\textwidth}\label{B1}
\includegraphics[width=1\textwidth,height=2.2cm]{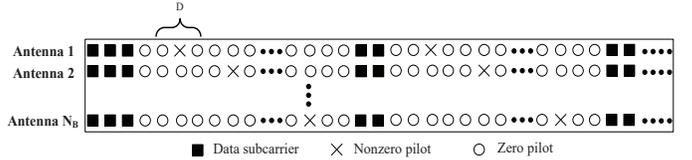}
\end{minipage}}

\caption{Comparison of different pilot patterns.}
\end{figure}

To clearly elaborate the motivation of our designed pilot pattern, we introduce two existing pilot patterns as shown in Fig. \ref{A1} and Fig. \ref{A2}.
\begin{itemize}
  \item In Fig. \ref{A1}, we present a pilot pattern proposed for DS channels in SISO system  \cite{Cheng2013}.  Assume that the pilots are arranged in $G$ groups. For each group, a nonzero pilot  is equipped with $D-1$  zero guard pilots on each side. It selects the $D$ central pilots of each group for channel estimation and obtains an ICI-free structure.
  \item Fig. \ref{A2} depicts a superimposed pilot pattern, which is proposed  for frequency selective channels in large-scale MIMO systems \cite{Gao2014}. The nonzero pilots occupy the same positions for each antenna to reduce the large pilot overhead brought by the increased number of antennas. The sequence of the nonzero pilots for each antenna  consists of random ${\pm 1}$ and  different antennas are distinguished by different sequences. In another word, it utilizes the differentiation of different antennas in code domain.
\end{itemize}
Considering that our pilot pattern is designed for DS channels in large-scale MIMO systems, we combines the above two properties together. As shown in Fig. \ref{B2}, the guard pilots are inserted in the superimposed structure. It gives the consideration to both the ICI avoidance and the reduction of the pilot overhead. To our best knowledge, little is concerned about the pilot pattern for DS channels in large-scale MIMO systems.  In the existing literature \cite{Tang2011}, the orthogonal pattern with guard pilots, as shown in Fig. \ref{B1},  is prepared for DS channels in MIMO systems.  It cannot support more antennas since the requirement of pilot subcarriers is too large.

\subsection{The Proposed Estimator}

In this part, we propose a BDCS based estimator for DS channels in large-scale MIMO systems by  properly extending the DCS based channel estimation scheme in SISO systems.

\subsubsection{Review of the DCS based estimator in \cite{Cheng2013}}

The work \cite{Cheng2013} proposed a DCS based estimator for DS channels in SISO systems. Its pilot pattern is depicted as Fig. \ref{A1}. The pilots are divided into $G$ groups and for each group,  the central $D$ pilots are selected for channel estimation to guarantee the ICI-free structure.  Let  $\mathcal{S}_d$,  $d \in [0,D-1]$, represent the index set of the $(d+1)$-th selected pilot subcarriers of all the groups. The index set of the nonzero pilots $\mathcal{S}_{\frac{{D - 1}}{2}}$ is also denoted as $\mathcal{S}_{cen}$. Thus we have
 \begin{equation}
  \begin{array}{*{20}{c}}
{{{\mathcal{S}}_0} = {{\mathcal{S}}_{cen}} - \frac{{D - 1}}{2}},\\
 \vdots \\
{{{\mathcal{S}}_{\frac{{D - 1}}{2}}} = {{\mathcal{S}}_{cen}}},\\
 \vdots \\
{{{\mathcal{S}}_{D - 1}} = {{\mathcal{S}}_{cen}} + \frac{{D + 1}}{2}}.
\end{array}
\end{equation}
Let ${\mathbf{P}}^{SISO} \in  {\mathbb{C}}^{G \times 1}$ denote the values of the nonzero pilots. The channel estimator, which describes the relationship between the received pilots and the CE-BEM coefficients, is derived as Lemma \ref{review}.
\begin{lemm}\label{review}
The CE-BEM coefficients $\tilde{\bm{\theta}}^{SISO}_d$, $d \in [0,D-1]$ can be obtained  by solving
\begin{equation}\label{28}
\left\{ {\begin{array}{*{20}{c}}
{{{\left[ {\mathbf{Y}^{SISO}} \right]}_{{{\mathcal S}_0}}} = {\tilde{\mathbf P}}^{SISO}{{\left[ {{{\mathbf{W}}_L}} \right]}_{{{\mathcal S}_{\frac{{D - 1}}{2}}}}}{\tilde{\bm{\theta}}^{SISO}_0} + {{\bm{\eta}}^{SISO}_0}}\\
 \vdots \\
{{{\left[ {\mathbf{Y}^{SISO}} \right]}_{{{\mathcal S}_{\frac{{D - 1}}{2}}}}} = {\tilde{\mathbf P}}^{SISO}{{\left[ {{{\mathbf{W}}_L}} \right]}_{{{\mathcal S}_{\frac{{D - 1}}{2}}}}}{\tilde{\bm{\theta}}^{SISO} _{\frac{{D - 1}}{2}}} + {{\bm{\eta}}^{SISO} _{\frac{{D - 1}}{2}}}}\\
 \vdots \\
{{{\left[ {\mathbf{Y}^{SISO}} \right]}_{{{\mathcal S}_{D - 1}}}} = {\tilde{\mathbf P}}^{SISO}{{\left[ {{{\mathbf{W}}_L}} \right]}_{{{\mathcal S}_{\frac{{D - 1}}{2}}}}}{\tilde{\bm{\theta}}^{SISO}_{D - 1}} + {{\bm{\eta}}^{SISO}_{D - 1}}}
\end{array}} \right.
\end{equation}
Here, ${\mathbf{Y}^{SISO}}$ represents the received signal and ${{\left[ {\mathbf{Y}^{SISO}} \right]}_{{{\mathcal S}_{d}}}}$ means the received pilot subcarriers, $d \in [0,D-1]$. ${\tilde{\mathbf P}}^{SISO}{=}diag{({\mathbf{P}}^{SISO})}$.
${\tilde{\bm{\theta}}^{SISO}_{d}}$ is the BEM coefficients with order $d$, $d \in [0,D-1]$. ${{\bm{\eta}}^{SISO}_{d}}$ is the noise term, $d \in [0,D-1]$.
\end{lemm}
\begin{proof}
The proof is summarized  briefly as follows. The received signal in frequency domain is
\begin{equation}
{\mathbf{Y}^{SISO}} = {\mathbf{H}^{SISO}_f}{\mathbf{S}^{SISO}} + {\bm{\Upsilon}}^{SISO},
\end{equation}
in which, ${\mathbf{S}^{SISO}}$ is the transmitted signal and  ${\bm{\Upsilon}}^{SISO}$ is the noise term.
Combining with the formula (\ref{matfre}), the CE-BEM decomposition of the received signal is derived as
\begin{equation}
\begin{split}
{\mathbf{Y}^{SISO}} &= \left( \sum\limits_{d = 0}^{D - 1} {{{\mathbf{V}}_d}{{\bm{\Theta}}^{SISO}_d}} \right){\mathbf{S}^{SISO}} + {\bm{\Upsilon}}^{SISO}    \\
&= \left( {\sum\limits_{d = 0}^{D - 1} {\mathbf{I}_N^{ < d - \frac{{D - 1}}{2} > }{{\tilde{\mathbf S}}^{SISO}}{\mathbf{W}_L}{{\tilde{\bm{\theta}}^{SISO}_{d}}}} } \right) + {\bm{\Upsilon}}^{SISO},
\end{split}
\end{equation}
in which, ${\tilde{\mathbf{S}}^{SISO}}{=}diag{({{\mathbf{S}}^{SISO}} )}$ and $\mathbf{I}_N^{ < d - \frac{{D - 1}}{2} >}$ means that the identity matrix with order $N$ shifts down circularly for $d - \frac{{D - 1}}{2}$ rows.

Assume that ${\mathbf{\Psi}}_{d^{'}}=[\mathbf{I}_N]_{{{\mathcal S}_{d^{'}}}}$, $d^{'} \in [0,D-1]$, the received pilots can be expressed as
\begin{equation}\label{31}
[{\mathbf{Y}^{SISO}}]_{{{\mathcal S}_{d^{'}}}}{=} {\mathbf{\Psi}}_{d^{'}}\left( {\sum\limits_{d = 0}^{D - 1} {\mathbf{I}_N^{ < d - \frac{{D - 1}}{2} > }{{\tilde{\mathbf S}}^{SISO}}{\mathbf{W}_L}{{\tilde{\bm{\theta}}^{SISO}_{d}}}} } \right) + {{\bm{\eta}}^{SISO}_{d^{'}}}
\end{equation}
in which,
\begin{equation}\label{32}
{\mathbf{\Psi}_{d^{'}}}{\mathbf{I}_N^{ < d - \frac{{D - 1}}{2} >} }{{\tilde{\mathbf S}}^{SISO}} = \left\{ {\begin{array}{*{20}{c}}
{{\tilde{\mathbf P}}^{SISO}}{[\mathbf{I}_N]_{{\mathcal{S}}_{cen}}}\\
{\mathbf{0}}
\end{array}} \right.{\rm{ }}\begin{array}{*{20}{c}}
{d = d^{'}=\frac{{D - 1}}{2}} \\
{else}
\end{array}.
\end{equation}
Substituting (\ref{32}) into (\ref{31}), it is not hard to obtain (\ref{28}).
\end{proof}

Thus, it is found  that
\begin{equation}\label{dd}
\begin{split}
&\left[
{{{[{{\mathbf{Y}}^{SISO}}]}_{{{\mathcal S}_0}}}} \cdots {{{[{{\mathbf{Y}}^{SISO}}]}_{{{\mathcal S}_{\frac{{_{D - 1}}}{2}}}}}} \cdots {{{[{{\mathbf{Y}}^{SISO}}]}_{{{\mathcal S}_{D - 1}}}}} \right]  \\
&={{\tilde{\mathbf{P}}}^{SISO}}{[{{\mathbf{W}}_L}]_{{{\mathcal S}_{\frac{{_{D - 1}}}{2}}}}}\left[ {
{{\tilde {\bm{\theta }}}_0^{SISO}} \cdots {{\tilde {\bm{\theta }}}_{\frac{{_{D - 1}}}{2}}^{SISO}} \cdots {{\tilde {\bm{\theta }}}_{D - 1}^{SISO}}} \right]  \\
&+\left[ {
{{\bm{\eta }}_0^{SISO}} \cdots {{\bm{\eta }}_{\frac{{_{D - 1}}}{2}}^{SISO}} \cdots {{\bm{\eta }}_{D - 1}^{SISO}}} \right].
\end{split}
\end{equation}
As analyzed above, the elements of $\{{\tilde {\bm{\theta }}}_{d}^{SISO}\}$, $d \in [0, D-1]$ have common sparsity among different BEM orders.  Thus (\ref{dd}) is a typical DCS problem and can be solved by the SOMP algorithm.

\subsubsection{The proposed BDCS based estimator}

Firstly we introduce a key finding that the CE-BEM coefficients present both the common sparsity and the block sparsity. Then exploiting the structured sparsity,  we propose a BDCS based estimator for DS channel estimation in large-scale MIMO systems.

The pilot pattern of the large-scale MIMO system is described in Fig. \ref{B2}. For each transmit antenna, the positions of the pilot subcarriers are arranged as the same way as the SISO systems. Besides, among different antennas, we adopt the superimposed structure, in which, the pilot subcarriers of all the antennas have the same indexes.

Let $\mathbf{P}^{(n_B)} \in {\mathbb{C}}^{G \times 1}$ represent the pilot values of the $n_B$-th antenna, $n_B \in [1,N_B]$. Considering that the received signal is the addition of the signals transmitted by all the antennas on the base station, we derive the relationship between the received pilots and the CE-BEM coefficients of all the transmit-receive antenna pairs as Theorem \ref{theo2}.
\begin{theo}\label{theo2}
The estimation of the CE-BEM coefficients of all the transmit-receive antenna pairs $\{{\tilde{\bm \theta }}_{d}^{({n_B})}\}$, $d {\in} [0,D-1]$, $n_B  {\in} [1,N_B]$, is conducted by (\ref{con2}), in which, ${{\tilde{\mathbf P}}^{({n_B})}} {=} diag{( {{\mathbf{P}}^{({n_B})}})}$.
\end{theo}
\begin{figure*}
\begin{equation}\label{con2}
\left\{ {\begin{array}{*{20}{c}}
{{{\left[ {\mathbf Y} \right]}_{{ {\mathcal S}_0}}} = \left[ {
{{\tilde {\mathbf P}}^{{(1)}}  {{[{\mathbf{W}_L}]}_{{{\mathcal S}_{\frac{{D - 1}}{2}}}}}} \cdots {{\tilde {\mathbf P}}^{{(N_B)}}{{[{\mathbf{W}_L}]}_{{{\mathcal S}_{\frac{{D - 1}}{2}}}}}}
} \right]{\left[{
{{\tilde{\bm{\theta}}} {{_0^{(1)}}^T}}} \cdots{
{{\tilde{\bm{\theta}}} {{_0^{(N_B)}}^T}}}\right]}^T} + {{\bm \eta} _0}\\
 \vdots \\
{{{\left[ {\mathbf Y} \right]}_{{{\mathcal S}_{\frac{{D - 1}}{2}}}}} = \left[ {
{{\tilde {\mathbf P}}^{{(1)}}{{[{{\mathbf{W}_L}}]}_{{{\mathcal S}_{\frac{{D - 1}}{2}}}}}} \cdots {{\tilde {\mathbf P}}^{{(N_B)}}{{[{{\mathbf{W}_L}}]}_{{{\mathcal S}_{\frac{{D - 1}}{2}}}}}}
} \right]{{\left[ {
{{{\tilde{\bm{\theta}}} {{_{\frac{{D - 1}}{2}}^{(1)}}^T}}}
\cdots {{{\tilde{\bm{\theta}}} {{_{\frac{{D - 1}}{2}}^{(N_B)}}^T}}}} \right]}^T} + {{\bm \eta}_{\frac{{D - 1}}{2}}}}\\
 \vdots \\
{{{\left[ {\mathbf Y} \right]}_{{{\mathcal S}_{D - 1}}}} = \left[{
{{\tilde {\mathbf P}}^{{(1)}}{{[{{\mathbf{W}_L}}]}_{{{\mathcal S}_{\frac{{D - 1}}{2}}}}}} \cdots{{\tilde {\mathbf P}}^{{(N_B)}}{{[{{\mathbf{W}_L}}]}_{{{\mathcal S}_{\frac{{D - 1}}{2}}}}}}} \right]{{\left[ {
{{{\tilde{\bm{\theta}}} {{_{D-1}^{(1)}}^T}}} \cdots{{\tilde{\bm{\theta}}} {{_{D-1}^{(N_B)}}^T}}} \right]}^T} + {{\bm \eta} _{D-1}}}
\end{array}} \right.
\end{equation}
\end{figure*}
\begin{proof}
Considering that the transmitted signals of different antennas overlap with each other in the air, we have that the received signal is
\begin{equation}
\begin{split}
{\mathbf{Y}} &= \sum\limits_{{n_B} = 1}^{{N_B}} {{\mathbf{H}}_f^{({n_B})}{{\mathbf{S}}^{({n_B})}} + \bm{\Upsilon} } \\
 &= \sum\limits_{{n_B} = 1}^{{N_B}} {\left( {\sum\limits_{d = 0}^{D - 1} {{{\mathbf{V}}_d}{\bm \Theta} _d^{({n_B})}} } \right){{\mathbf{S}}^{({n_B})}} + \bm{\Upsilon} } \\
 &= \sum\limits_{{n_B} = 1}^{{N_B}} {\left( {\sum\limits_{d = 0}^{D - 1} {{\mathbf{I}}_N^{ < d - \frac{{D - 1}}{2} > }{{\tilde {\mathbf{S}}}^{({n_B})}}{{\mathbf{W}}_L}
 {{\tilde {\bm \theta}}_d^{({n_B})}}} } \right) + \bm{\Upsilon} },
\end{split}
\end{equation}
in which, ${\mathbf{S}}^{(n_B)}$ is the transmitted signal of the $n_B$-th antenna,  ${{\tilde{\mathbf S}}^{({n_B})}} {=} diag{( {{\mathbf{S}}^{({n_B})}})}$,  $n_B \in [1,N_B]$, and ${\bm{\Upsilon}}$ is the noise term.

Then the pilot subcarriers are selected as
\begin{equation}
\begin{split}
{[{\mathbf{Y}}]_{{S_{{d^{'}}}}}} &= {{\mathbf{\Psi }}_{{d^{'}}}}\sum\limits_{{n_B} = 1}^{{N_B}} {\left( {\sum\limits_{d = 0}^{D - 1} {{\mathbf{I}}_N^{ < d - \frac{{D - 1}}{2} > }{{\tilde {\mathbf{S}}}^{({n_B})}}{{\mathbf{W}}_L}\tilde {\bm \theta} _d^{({n_B})}} } \right) + {{\bm \eta} _{{d^{'}}}}} \\
 &= \sum\limits_{{n_B} = 1}^{{N_B}} {{{\mathbf{\Psi }}_{{d^{'}}}}\left( {\sum\limits_{d = 0}^{D - 1} {{\mathbf{I}}_N^{ < d - \frac{{D - 1}}{2} > }{{\tilde {\mathbf{S}}}^{({n_B})}}{{\mathbf{W}}_L}\tilde {\bm \theta} _d^{({n_B})}} } \right) + {{\bm \eta} _{{d^{'}}}}} \\
 &= \sum\limits_{{n_B} = 1}^{{N_B}} {{{\tilde {\mathbf{P}}}^{({n_B})}}{{[{{\mathbf{W}}_L}]}_{{S_{\frac{{D - 1}}{2}}}}}\tilde {\bm \theta} _{{d^{'}}}^{({n_B})} + {{\bm \eta} _{{d^{'}}}}},
\end{split}
\end{equation}
in which $d^{'} \in [0,D-1]$. Thus  (\ref{con2}) is obtained as the complete expression.
\end{proof}
\begin{figure}[!t]
\centering
\includegraphics[scale=0.4]{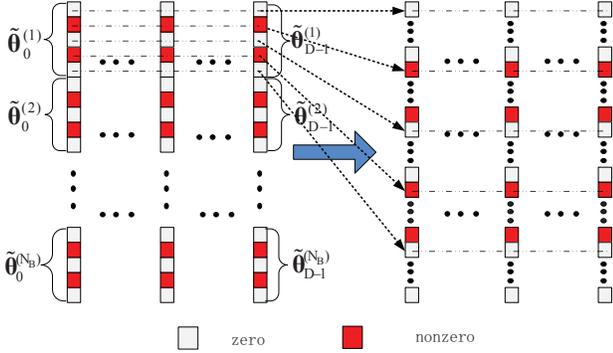}
\caption{The diagram of the common and block sparsity with $L=5$ and $K=2$.}
\label{bspa}
\end{figure}

For convenience of observation, we reshape (\ref{con2}) in a compact form as
\begin{equation}\label{compact}
\begin{split}
&\left[ {{{{[{\mathbf{Y}}]}_{{{\mathcal S}_0}}}} \cdots {{{[{\mathbf{Y}}]}_{{{\mathcal S}_{\frac{{D - 1}}{2}}}}}} \cdots {{{[{\mathbf{Y}}]}_{{{\mathcal S}_{D - 1}}}}}} \right] = {\mathbf{Z}}{\mathbf{\Lambda }}\\
&+ \left[ {{{{[{\bm \eta} ]}_{{{\mathcal S}_0}}}} \cdots {{{[{\bm \eta} ]}_{{{\mathcal S}_{\frac{{D - 1}}{2}}}}}} \cdots {{{[{\bm \eta} ]}_{{{\mathcal S}_{D - 1}}}}}} \right],
\end{split}
\end{equation}
in which,
\begin{equation}\label{38}
{\mathbf{Z}} = \left[ {{{\tilde {\mathbf{P}}}^{(1)}}{{[{{\mathbf{W}}_L}]}_{{{\mathcal S}_{\frac{{D - 1}}{2}}}}} \cdots {{\tilde {\mathbf{P}}}^{({N_B})}}{{[{{\mathbf{W}}_L}]}_{{{\mathcal S}_{\frac{{D - 1}}{2}}}}}} \right],
\end{equation}
and
\begin{equation}
{\mathbf{\Lambda }}=\left[ {\begin{array}{*{20}{c}}
{\tilde {\bm \theta} _0^{(1)}}& \cdots &{\tilde {\bm \theta} _{D - 1}^{(1)}}\\
 \vdots & \ddots & \vdots \\
{\tilde {\bm \theta} _0^{({N_B})}}& \cdots &{\tilde {\bm \theta} _{D - 1}^{({N_B})}}
\end{array}} \right].
\end{equation}
For better illustration, Fig. \ref{bspa} depicts an example of the CE-BEM coefficient matrix ${\mathbf{\Lambda}}$ with the channel length $L=5$ and the sparsity $K=2$. It is obvious that ${\mathbf{\Lambda}}$ has common sparsity among its columns. Besides, when we adjust the rows of ${\mathbf{\Lambda}}$ in the order of the channel tap, it presents the nonzero blocks and the zero blocks, referred to as the block sparsity. It can be concluded that the CE-BEM coefficient matrix ${\mathbf{\Lambda}}$ have both common and block sparsity since the elements of $\{{\tilde {\bm \theta} _{d}^{(n_B)}}\}$, $d \in [0,D-1]$, $n_B \in [1,N_B]$ have common sparsity not only  among different CE-BEM orders, but also among different antennas.
Now we can see that the proposed estimator (\ref{con2}) attributes to a BDCS problem, which can be solved by SSP or BSOMP algorithms.

 The work \cite{Dai2013} proposed an estimator for DS channels in large-scale MIMO systems as well, which can only be performed in time frequency training (TFT)-OFDM systems.  It utilizes the correlation of the  Pseudo-random sequence for the nonzero elements localization and is not suitable for  the conventional cyclic prefix (CP)-OFDM.

\subsection{The Proposed Pilot Design Algorithm}

In this part, we propose a novel pilot design algorithm to optimize the pilot positions ${\mathcal S}_{cen}$ for  better performance of estimation. Considering the block structure, it decreases the coherence among different blocks of the measurement matrix, which is different from the MCP of the pilot design algorithm for SISO systems.

As analyzed above, the CE-BEM coefficient matrix ${\mathbf{\Lambda }}$ presents block sparsity when its rows are arranged in the order of the channel tap. For more accurate localization, the columns of the measurement matrix ${\mathbf{Z}}$ are also arranged in the
order of the channel tap accordingly. Assume that ${\mathcal A}_l=\{l+1,\, l+1+L,\, l+1+2L,\,\cdots, l+1+(N_B-1)L\}$, $l \in [0,L-1]$, and we have
\begin{equation}\label{40}
{{\mathbf{Z}}_{a}} = \left[ {
{{{[{{\mathbf{Z}}^T}]}_{{{\mathcal A}_0}}}^T},{{{[{{\mathbf{Z}}^T}]}_{{{\mathcal A}_1}}}^T,} \cdots, {{{[{{\mathbf{Z}}^T}]}_{{{\mathcal A}_{L-1}}}}^T}} \right].
\end{equation}
Apparently, $\mathbf{Z}_a$ consists of $L$ blocks and the coherence among them affects the localization of the nonzero blocks. Assume that
\begin{equation}\label{41}
{{\mathbf{Z}}_s} = \left[ {
{vec({{[{{\mathbf{Z}}^T}]}_{{{\mathcal A}_0}}}^T)},{vec({{[{{\mathbf{Z}}^T}]}_{{{\mathcal A}_1}}}^T)}, \cdots, {vec({{[{{\mathbf{Z}}^T}]}_{{{\mathcal A}_{L-1}}}}^T)}
} \right],
\end{equation}
and $\mu(\mathbf{Z}_s)$ takes on the reference value of the coherence among difference blocks. In our proposed Algorithm \ref{AA1}, referred to as BDSO, we search for a proper nonzero pilot positions  ${\mathcal S}_{\frac{D-1}{2}}$ in (\ref{38}) with a small $\mu(\mathbf{Z}_s)$ by the iteration process.  The detailed steps are summarized as follows.

\begin{algorithm}[t]
\caption{ The proposed pilot design  algorithm: BDSO}
\label{AA1}
\textbf{Initialization}:
\begin{algorithmic}[1]

\STATE
 Generate  random $\pm 1$ sequences for ${\mathbf{P}}^{(n_B)} \in {\mathbb{Z}^{G \times 1}}$, $n_B \in [1,N_B]$.
\STATE
 Generate equidistant indexes $\widetilde{\mathcal{S}}_{cen}^{(0)}=\{{\widetilde{s}}^{(0)}_1, {\widetilde{s}}^{(0)}_2, \ldots, {\widetilde{s}}^{(0)}_G \}$, which satisfies ${\widetilde{s}}^{(0)}_u \in [1,N], \left| {{\widetilde{s}}_u^{(0)} - {\widetilde{s}}_v^{(0)}} \right| {\ge} 2D-1$, $u, v \in [1,G]$.
\STATE
${{\bm{\rho }}_0} = {\mathbf{0}}$,  ${{{\rho }}_{0,0}} = 1$,   $i=j=0$, and ${\mathcal{S}}_{cen}^{(0)} = \widetilde{\mathcal{S}}_{cen}^{(0)}$.
\end{algorithmic}
\emph{for} $m = 1, \ldots, Iter$
\begin{enumerate}[1)]
\item \textbf{Reformulation}

 Obtain $\bar {\mathcal S}_{cen}^{(m - 1)}$ by changing an element of $\widetilde {\mathcal{S}}_{cen}^{(m-1)}$, $\left| {\bar s_u^{(m - 1)} - \bar s_v^{(m - 1)}} \right| \ge 2D - 1$ and then reformulate the measurement matrix ${{\bar{{\mathbf {Z}} }}^{(m - 1)}}$ and ${{\widetilde{{\mathbf {Z}} }}^{(m - 1)}}$ according to $\bar {\mathcal S}_{cen}^{(m - 1)}$ and $\widetilde {\mathcal S}_{cen}^{(m - 1)}$ as (\ref{38}).

\item \textbf{Conversion}

 Convert ${{\bar{{\mathbf {Z}} }}^{(m - 1)}}$ and ${{\widetilde{\mathbf{{ {Z}} }}}^{(m - 1)}}$ to ${{\bar{{\mathbf {Z}} }}_s^{(m - 1)}}$ and ${{\widetilde{\mathbf{{ {Z}} }}}_s^{(m - 1)}}$ as (\ref{40}) and (\ref{41}).

\item \textbf{Calculation}

\emph{if} ${\mu}(\bar {\mathbf{Z}}^{(m-1)}_{s})  < {\mu}({\widetilde{\mathbf{{Z} }}}^{(m-1)}_{s})$,

    \quad set $\widetilde {\mathcal{S}}_{cen}^{(m)}=\bar {\mathcal S}_{cen}^{(m - 1)}$, $i=m+1$;

\emph{else}

     \quad set  $\widetilde {\mathcal{S}}_{cen}^{(m)}= \widetilde {\mathcal{S}}_{cen}^{(m-1)}$.

\emph{end}
\item \textbf{Probability}

${{\bm{\rho }}_m} = {{\bm{\rho }}_{m - 1}} + \frac{1}{m}({\mathbf{T}_i} - {{\bm{\rho }}_{m - 1}})$, in which ${\mathbf{T}_i}$ is a zero vector except that the $i$-th element is 1.
\item \textbf{Update}

\emph{if} ${{\rho }_{m,i}} > {{\rho }_{m,j}}$,

\quad ${\mathcal{S}}_{cen}^{(m)} = \widetilde{\mathcal{S}}_{cen}^{(m)}$, $j=i$;

\emph{else}

\quad ${\mathcal{S}}_{cen}^{(m)} =  {\mathcal{S}}_{cen}^{(m - 1)}$.

\emph{end}
\end{enumerate}
\emph{end}
\end{algorithm}

We initialize the values of the nonzero pilots ${\mathbf{P}}^{(n_B)} \in {\mathbb{Z}^{G \times 1}}$, $n_B \in [1,N_B]$ with random sequences  consisting of $\pm 1$, and then generate  ${\widetilde {\mathcal{S}}_{cen}}=\{{\widetilde{s}}_1, {\widetilde{s}}_2, \ldots, {\widetilde{s}}_G \}$  randomly, which satisfies $\left| {{\widetilde{s}}_u - {\widetilde{s}}_v} \right| \ge 2D - 1, {\widetilde{s}}_u, {\widetilde{s}}_v \in [0,N-1], u, v \in [1,G]$. The condition guarantees the groups of pilots do not overlap with each other.
In the reformulation stage, ${{\bar{\mathcal{S}}}_{cen}}$ is obtained by changing an element of ${\widetilde {\mathcal{S}}_{cen}}$ and  ${\bar{\mathbf{Z}}}$ and ${\widetilde{\mathbf{Z}}}$  are reformulated accordingly as (\ref{38}). Then we convert the form of the matrix   ${\bar{\mathbf{Z}}}$ and ${\widetilde{\mathbf{Z}}}$ to ${\bar{\mathbf{Z}}}_{s}$  and ${\widetilde{\mathbf{Z}}}_{s}$  as (\ref{40}) and (\ref{41}) for the convenience of calculating the coherence. Thus the ${\widetilde {\mathcal{S}}_{cen}}$ is updated  according to  the  comparison between ${{\mu}}(\bar {\mathbf{Z}}_{s})$ and  ${{\mu}}({\widetilde{\mathbf{Z}}}_{s})$. Finally,  in the state transition stage,  the probability of each state is calculated. Comparing the current and the last states, the ${\mathcal{S}_{cen}}$ which corresponds to the larger probability is selected.  As the iteration process converges,   we obtain the optimized ${\mathcal{S}_{cen}}$ with a steady solution.

In contrast to the pilot design algorithm \cite{Cheng2013} proposed for the SISO systems, referred to as  DSO, we pay attention to the coherence among different blocks of the measurement matrix instead of the columns, which avoids a large amount of unnecessary search. The works \cite{Qi2015} and \cite{He2015} proposed the pilot design algorithms for MIMO systems as well, which utilized the genetic and the iterative search to optimize the pilot positions. However, they were designed for the orthogonal pilot patterns  and  assigned different pilot positions to different antennas. Thus, a large number of pilot subcarriers were required and the estimator could only support a maximum of 4 antennas.

\subsection{Linear Smoothing}

In order to relieve the performance deterioration brought by the modeling error of the CE-BEM, our previous work \cite{Qin2016} proposed a linear smoothing method for SISO systems. In specific, we approximate the channel coefficients corresponding to the $N$ time instants with a straight line. Here we extend it to large-scale MIMO systems by repeating the smoothing process for each transmit-receive antenna pair.
\begin{itemize}
  \item Select the indices of the nonzero taps $\left\{ {l_1}, \cdots ,{l_K} \right\}$ in $[0,L-1]$ by comparison of  the  channel coefficients ${\mathbf{h}}_{l}^{({n_B})}$, $l \in [0,L-1]$, $n_B \in [1,N_B]$.
  \item Obtain the estimated channel coefficients corresponding to the two central instants of each $l_k$, $k \in [1,K]$.
\begin{align}
&\hat h^{({n_B})}[\frac{N}{4}-1,l_k] {\approx} \frac{2}{N}(\sum\limits_{n = 0}^{N/2 - 1} {h^{({n_B})}[n,{l_k}]} ),    \notag \\
&\hat h^{({n_B})}[\frac{3N}{4}-1,{l_k}] {\approx} \frac{2}{N}(\sum\limits_{n = N/2}^{N-1} {h^{({n_B})}[n,{l_k}]} ),    \notag\\
&l_k \in \left\{ {l_1}, \cdots ,{l_K} \right\}, n_B \in [1,N_B].   \notag
\end{align}
  \item Calculate the slope of the approximate line determined by the two central points of each ${\mathbf{h}}_{l_k}^{({n_B})}$:
\begin{align}
&\beta _{{l_k}}^{({n_B})} = \frac{{\hat h^{({n_B})}[\frac{N}{4}-1,{{l_k}}] - \hat h^{({n_B})}[\frac{3N}{4}-1,{{l_k}}]}}{{N/2}}, \notag  \\
&l_k \in \left\{ {l_1}, \cdots ,{l_K} \right\}, n_B \in [1,N_B].  \notag
\end{align}
  \item The processed channel can be expressed as
\begin{align}
&\hat{h}^{({n_B})}[n,{l_k}] = (n + 1 - \frac{N}{4})\beta _{{l_k}}^{({n_B})} + \hat h^{({n_B})}[\frac{N}{4}-1,{l_k}],   \notag \\
&l_k \in \left\{ {l_1}, \cdots ,{l_K} \right\}, n_B \in [1,N_B], n \in [0,N-1].    \notag
\end{align}
\end{itemize}

\emph{Remark:}
As described above, we mainly focus on the discussion about the downlink DS channel estimation in the large-scale systems, which  is a challenging problem involving the IAI. For the integrity of the work, the uplink DS channel estimation in the large-scale MIMO systems is briefly introduced,  although it is relatively simple and not concerned with the IAI. The estimator can be expressed as (\ref{up}), in which, ${{\mathbf{Y}}^{'}}^{(n_B)}$ represents the signal received by the $n_B$-th antenna on the base station, ${\tilde {\mathbf{P}}}^{'}{=}diag {({{\mathbf{P}}^{'}})}$,  ${{\mathbf{P}}^{'}} {\in} {\mathbb{Z}}^{G \times 1}$ denotes the pilot sequence, ${{{\tilde {\bm \theta }}}^{'}}{_{d}^{(n_B)}}$ is the CE-BEM coefficients and ${{{\tilde {\bm \eta }}}^{'}}{_{d}^{(n_B)}}$ is the noise term, $d \in [0,D-1]$, $n_B \in [1,N_B]$. It is found that the CE-BEM coefficients   presents  common sparsity while without block sparsity, which is different from the downlink channel estimator. The uplink channel estimator (\ref{up}) attributes to a DCS problem, while the downlink channel estimator (\ref{con2}) is a BDCS problem.  Thus the recovery algorithm SOMP is suitable for the problem, instead of SSP.
\begin{figure*}
\begin{equation}\label{up}
\begin{split}
&\left[ {{{[{{\mathbf{Y}}^{'}}^{(1)}]}_{{{\mathcal S}_0}}}...{{[{{\mathbf{Y}}^{'}}^{(1)}]}_{{{\mathcal S}_{\frac{{D - 1}}{2}}}}}...{{[{{\mathbf{Y}}^{'}}^{(1)}]}_{{{\mathcal S}_{D - 1}}}}...{{[{{\mathbf{Y}}^{'}}^{({N_B})}]}_{{{\mathcal S}_0}}}...{{[{{\mathbf{Y}}^{'}}^{({N_B})}]}_{{{\mathcal S}_{\frac{{D - 1}}{2}}}}}...{{[{{\mathbf{Y}}^{'}}^{({N_B})}]}_{{{\mathcal S}_{D - 1}}}}} \right]\\
& = {{\tilde {\mathbf P}}^{'}}{[{{\mathbf{W}}_L}]_{{{\mathcal S}_{\frac{{D - 1}}{2}}}}}\left[ {{{{\tilde{\bm \theta }}}^{'}}{_0^{(1)}}...{{{\tilde{\bm \theta }}}^{'}}{_{\frac{{D - 1}}{2}}^{(1)}}...{{{\tilde{\bm \theta }}}^{'}}{_{D - 1}^{(1)}}...{{{\tilde{\bm \theta }}}^{'}}{_0^{({N_B})}}...{{{\tilde{\bm \theta }}}^{'}}{_{\frac{{D - 1}}{2}}^{({N_B})}}...{{{\tilde{\bm \theta }}}^{'}}{_{D - 1}^{({N_B})}}} \right]
 + \left[ {{{{\tilde{\bm \eta }}}^{'}}{_0^{(1)}}...{{{\tilde{\bm \eta }}}^{'}}{_{\frac{{D - 1}}{2}}^{(1)}}...{{{\tilde{\bm \eta }}}^{'}}{_{D - 1}^{(1)}}...{{{\tilde{\bm \eta }}}^{'}}{_0^{({N_B})}}...{{{\tilde{\bm \eta }}}^{'}}{_{\frac{{D - 1}}{2}}^{({N_B})}}...{{{\tilde{\bm \eta }}}^{'}}{_{D - 1}^{({N_B})}}} \right]
\end{split}
\end{equation}
\end{figure*}

\section{Complexity}

In this section, we conduct the analysis of the complexity of our proposed DS channel estimation scheme for the large-scale MIMO systems. The computation complexity overhead  consists of three parts, the recovery algorithm of the BDCS-based problem, the BDSO pilot design algorithm and the linear smoothing method. The complexity for the BSOMP algorithm is ${\mathcal O}(G^2)$, which has been analyzed in \cite{Qin2016}. Then as to the proposed BDSO pilot design algorithm,  the total times of multiplication for calculation of $\mu$ is $G{L^2}{N_B}$ and the complexity is ${\mathcal{O}}(IterG{L^2}{N_B})$.  Moreover, the complexity for the linear smoothing method is ${\mathcal{O}}(NKN_B)$. Finally, it is obtained that the complexity for our proposed channel estimation scheme is ${\mathcal O}(G^2+IterG{L^2}{N_B}+NKN_B)$.

\section{Simulation}

\begin{figure}[t]
\centering
\includegraphics[width=0.5\textwidth]{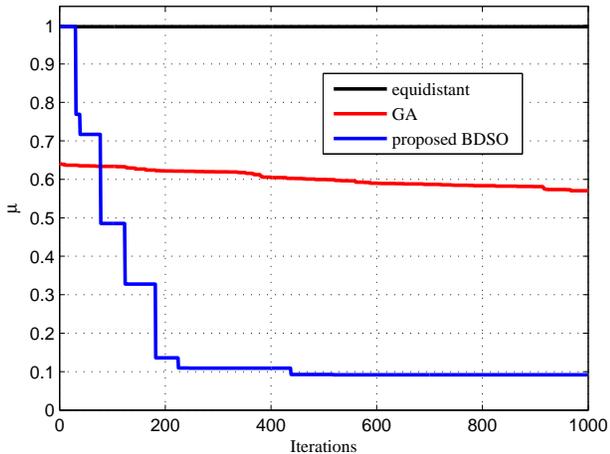}
\caption{The decrease of $\mu$ with different pilot design algorithms.}
\label{F5}
\end{figure}
\begin{figure}[t]
\centering
\includegraphics[width=0.5\textwidth]{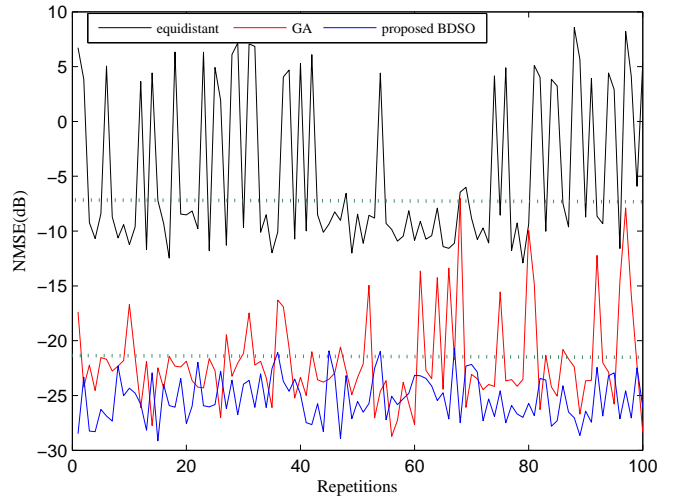}
\caption{The NMSE performance of the  channel estimator with different pilot design algorithms for 100 repetitions.}
\label{F6}
\end{figure}
\begin{figure}[t]
\centering
\includegraphics[width=0.5\textwidth]{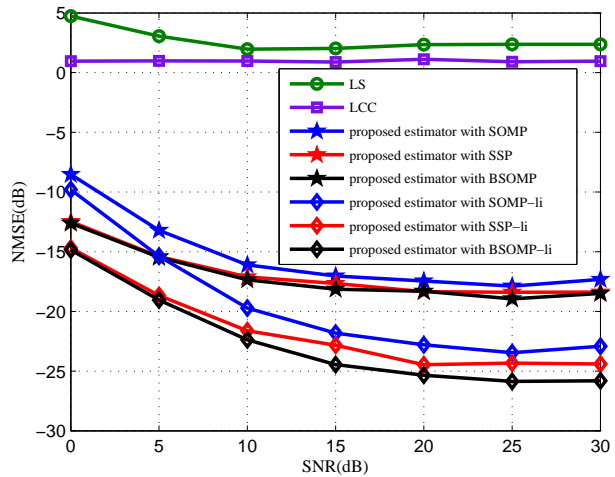}
\caption{The NMSE performance versus the SNR with the number of the antennas $N_B=12$ and the vehicular speed v=300km/h.}
\label{F7}
\end{figure}
\begin{figure}[t]
\centering
\includegraphics[width=0.5\textwidth]{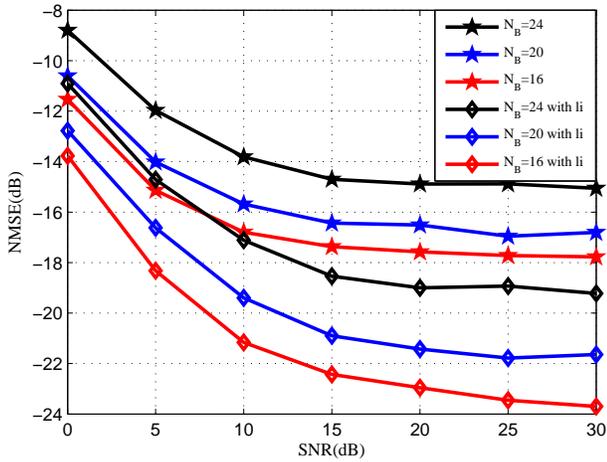}
\caption{The NMSE performance of our proposed BDCS based estimator versus the SNR with different number of antennas and the vehicular speed v=300km/h.}
\label{F8}
\end{figure}
\begin{figure}[t]
\centering
\includegraphics[width=0.5\textwidth]{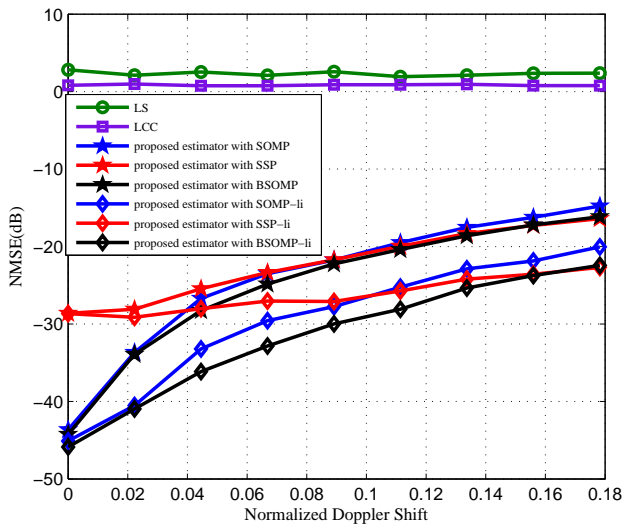}
\caption{The NMSE performance versus the normalized Doppler shift with SNR=30dB and the number of antennas $N_B=12$.}
\label{F9}
\end{figure}
\begin{figure}[t]
\centering
\includegraphics[width=0.5\textwidth]{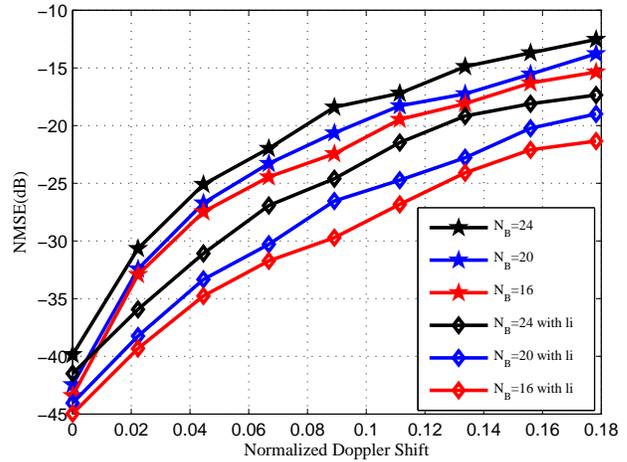}
\caption{The NMSE performance of our proposed BDCS based estimator versus the SNR with different number of antennas and SNR=30dB.}
\label{F10}
\end{figure}
\begin{figure}[t]
\centering
\includegraphics[width=0.5\textwidth]{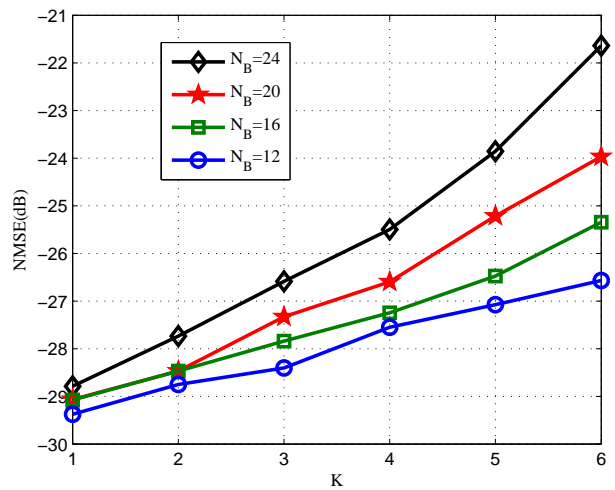}
\caption{The NMSE performance of our proposed BDCS based estimator versus the sparsity $K$ with SNR=25dB and the vehicular speed v=300km/h.}
\label{F11}
\end{figure}

In this section, we conduct simulations by MATLAB to verify the performance gain of our proposed BDCS based channel estimator and the proposed pilot design algorithm BDSO. The  proposed channel estimator is compared with the channel estimation schemes,  least squares (LS) and LCC \cite{Ren2013}. The LS algorithm performs pseudo-inverse directly, which ignores the sparsity of the CE-BEM coefficients and the LCC scheme is CS-based, which utilizes the sparsity in the delay-doppler domain. The recovery algorithms includes SOMP, SSP and BSOMP. As to the proposed pilot design algorithm BDSO, the equidistant pilot pattern \cite{Gao2015} and the genetic algorithm (GA) \cite{He2015}  are selected as the comparison schemes. The equidistant pilot pattern means that the nonzero pilots are distributed uniformly, which have a equal distant with each other and the GA algorithm exploits the genetic search to obtain the optimized pilot positions. The linear smoothing method is represented by ``li" in the figures.

The parameters of our simulated system are given here. We consider an OFDM system with the channel model ITU-Vehicular B \cite{Marques2010}, which is also adopted in the work \cite{Dai2013} for simulation. The number of the subcarriers $N=4096$ and the number of the nonzero ones  $G=192$.  We can see that the pilot overhead is $G(2D-1)/N=23.4\%$. The channel length is $L=200$ with $K=6$ nonzero paths.  The central frequency is 2.35GHz and the bandwidth is 20MHz. QPSK is selected as the modulation technique.

\subsection{The Verification of BDSO}

In Fig. \ref{F5}-\ref{F6}, we verify the effectiveness of our proposed pilot design algorithm BDSO under the number of the antennas $N_B{=}12$ and the vehicular speed  $v{=}300km/h$. Fig. \ref{F5} depicts the variation process of $\mu$ with the increasing of the iterations in Algorithm \ref{AA1}. We can see that the $\mu$ of the equidistant pilot pattern is approximately equal to 1, which is the upper bound of $\mu$. The curve corresponding to GA declines slightly in the iteration process with 1000 times.  The  convergence speed is too low and the final value of $\mu$ is approximately equal to 0.6. As to our proposed BDSO, in the former 500 iterations, it presents obvious trend of  declining and  in the latter 500 iterations, it reaches the steady state. The final value of $\mu$ is less than 0.1.

It is known that  two reference values can be utilized to measure the performance of the CS recovery algorithms. The one is the accuracy of the recovery and the other is the probability of the accurate recovery. Fig. \ref{F6} depicts the normalized mean square error (NMSE) performance of 100 repetitions for each pilot design algorithm with the recovery algorithm SSP. We can see that for the equidistant pilot pattern, the percentage of the NMSE skip points reaches almost $30\%$ and the amplitude of the NMSE skip point reaches 5dB.  Besides, the NMSE of the recovered points is just -10dB, which cannot satisfy the required accuracy of the estimation. The NMSE of the recovered points for the GA algorithm is -22dB, which is acceptable for the channel estimation. But the percentage of the NMSE skip points is $15\%$ and their amplitude reaches -15dB. As to the proposed algorithm BDSO, the NMSE  performance is -25dB  without any skip points. The results of the simulation are in accordance with the conclusion that a smaller $\mu$ leads to a better recovery performance.

\subsection{NMSE Versus SNR}

In Fig. \ref{F7},  we conduct the simulation to compare the performance of LS, LCC and our proposed BDCS based estimator under the number of the antennas $N_B=12$ and the vehicular  speed $v=300km/h$ with the BDSO optimized pilot positions. It is found that the curves of the LS and the LCC schemes are close to 0dB and they hardly vary with the increasing of SNR. They are ineffective under the above conditions since the pilot overhead cannot support so many antennas. We adopt three recovery algorithms to solve our proposed BDCS based estimator, including SOMP, SSP and BSOMP. The SSP and the BSOMP present similar accuracy of recovery and their corresponding curves almost coincide with each other. In contrast to SOMP, the SSP and BSOMP  have performance gain since they utilize both the block and the common sparsity instead of the common sparsity only. Besides, we perform the linear smoothing method combined with the three recovery algorithms and obtain the obvious performance gain. BSOMP-li presents  the best performance of the estimation. In Fig. \ref{F8}, we verify the effectiveness of our proposed BDCS based estimator with BSOMP under $N_B=16$, $N_B=20$, and $N_B=24$. It is found  that the performance deteriorates as the number of antennas increases since  their pilot overhead is the same.
\subsection{NMSE Versus Normalized Doppler Shift}

In Fig. \ref{F9}, we present the variation of the NMSE as the normalized Doppler shift increases under $SNR=30dB$ with the BSOMP algorithm. It is found  that the curve of the LS and LCC schemes are close to 0dB  since they are ineffective under the pilot overhead of 23.4\%. The recovery performance deteriorates with the normalized Doppler shift increasing since the increased ICI leads to the increased modeling error. The algorithm BSOMP has performance gain compared with the SOMP and  the gap is more obvious with linear smoothing method. The algorithm SSP presents the worst performance since it is sensitive to the interference.

In Fig. \ref{F10}, we verify the effectiveness of our proposed BDCS based estimator with more antennas $N_B=16$,  $N_B=20$,  $N_B=24$ utilizing the algorithm BSOMP  under $SNR=30dB$.
It is found  that the performance gets worse  as the number of antennas increases. All curves present the same rising  tendency  as the normalized Doppler shift increases. Besides, the linear smoothing method contributes to the performance gain.

\subsection{NMSE Versus the Sparsity $K$}

Fig. \ref{F11} depicts the variation of NMSE as the channel sparsity $K$ increases for our proposed BDCS based estimator under $SNR=25dB$ exploiting the algorithm BSOMP and the linear smoothing method. We present the simulation for $N_B=12$, $N_B=16$, $N_B=20$ and $N_B=24$. Since the pilot overhead is fixed,  the performance deteriorates as the sparsity increases under the same number of antennas and  the performance gets worse as the number of the antennas increases under the same  channel sparsity.

\section{Conclusion}

We propose a BDCS based DS channel estimator and the corresponding pilot design algorithm BDSO in large-scale MIMO systems. It supports more antennas and guarantees the performance of the estimation with the affordable pilot overhead. In the future, we may take account of the sparsity in beam domain and improve the performance of the estimation further.

\ifCLASSOPTIONcaptionsoff
  \newpage
\fi

\bibliographystyle{IEEEtran}
\bibliography{IEEEabrv,IEEEfull}

\begin{thebibliography}{10}
\providecommand{\url}[1]{#1}
\csname url@samestyle\endcsname
\providecommand{\newblock}{\relax}
\providecommand{\bibinfo}[2]{#2}
\providecommand{\BIBentrySTDinterwordspacing}{\spaceskip=0pt\relax}
\providecommand{\BIBentryALTinterwordstretchfactor}{4}
\providecommand{\BIBentryALTinterwordspacing}{\spaceskip=\fontdimen2\font plus
\BIBentryALTinterwordstretchfactor\fontdimen3\font minus
  \fontdimen4\font\relax}
\providecommand{\BIBforeignlanguage}[2]{{%
\expandafter\ifx\csname l@#1\endcsname\relax
\typeout{** WARNING: IEEEtran.bst: No hyphenation pattern has been}%
\typeout{** loaded for the language `#1'. Using the pattern for}%
\typeout{** the default language instead.}%
\else
\language=\csname l@#1\endcsname
\fi
#2}}
\providecommand{\BIBdecl}{\relax}
\BIBdecl

\bibitem{Bj??rnson2016}
E.~B. jrnson, E.~G. Larsson, and T.~L. Marzetta, ``Massive {MIMO}: ten myths
  and one critical question,'' \emph{IEEE Commun. Mag.}, vol.~54, no.~2, pp.
  114--123, Feb. 2016.

\bibitem{Lu2014}
L.~Lu, G.~Y. Li, A.~L. Swindlehurst, A.~Ashikhmin, and R.~Zhang, ``An overview
  of massive {MIMO}: Benefits and challenges,'' \emph{IEEE J. Sel. Topics
  Signal Process.}, vol.~8, no.~5, pp. 742--758, Oct. 2014.

\bibitem{Larsson2014}
E.~Larsson, O.~Edfors, F.~Tufvesson, and T.~Marzetta, ``Massive {MIMO} for next
  generation wireless systems,'' \emph{IEEE Commun. Mag.}, vol.~52, no.~2, pp.
  186--195, Feb. 2014.

\bibitem{Ren2015}
X.~Ren, W.~Chen, and M.~Tao, ``Position-based compressed channel estimation and
  pilot design for high-mobility {OFDM} systems,'' \emph{IEEE Trans. Veh.
  Technol.}, vol.~64, no.~5, pp. 1918--1929, May 2015.

\bibitem{Rangan2014}
S.~Rangan, T.~Rappaport, and E.~Erkip, ``Millimeter-wave cellular wireless
  networks: Potentials and challenges,'' \emph{Proc. IEEE}, vol. 102, no.~3,
  pp. 366--385, Mar. 2014.

\bibitem{Swindlehurst2014}
A.~L. Swindlehurst, E.~Ayanoglu, P.~Heydari, and F.~Capolino, ``Millimeter-wave
  massive {MIMO}: the next wireless revolution?'' \emph{IEEE Commun. Mag.},
  vol.~52, no.~9, pp. 56--62, Sept. 2014.

\bibitem{Ma2003}
X.~Ma, G.~Giannakis, and S.~Ohno, ``Optimal training for block transmissions
  over doubly selective wireless fading channels,'' \emph{IEEE Trans. on Signal
  Process.}, vol.~51, no.~5, pp. 1351--1366, May 2003.

\bibitem{Tang2007}
Z.~Tang, R.~C. Cannizzaro, G.~Leus, and P.~Banelli, ``Pilot-assisted
  time-varying channel estimation for {OFDM} systems,'' \emph{IEEE Trans. on
  Signal Process.}, vol.~55, no.~5, pp. 2226--2238, May 2007.

\bibitem{Tang2011}
Z.~Tang and G.~Leus, ``\BIBforeignlanguage{English}{Identifying time-varying
  channels with aid of pilots for {MIMO-OFDM}},''
  \emph{\BIBforeignlanguage{English}{EURASIP J. Advances in Signal Process.}},
  vol. 2011, no.~1, pp. 1--19, Sep. 2011.

\bibitem{Eldar2012}
Y.~C. Eldar and G.~Kutyniok, \emph{Compressed sensing: theory and
  applications}.\hskip 1em plus 0.5em minus 0.4em\relax Cambridge University
  Press, 2012.

\bibitem{Ren2013}
X.~Ren, W.~Chen, and Z.~Wang, ``Low coherence compressed channel estimation for
  high mobility {MIMO} {OFDM} systems,'' in \emph{Proc. 2013 IEEE GLOBECOM},
  pp. 3389--3393.

\bibitem{Gao2015}
Z.~Gao, L.~Dai, Z.~Wang, and S.~Chen, ``Spatially common sparsity based
  adaptive channel estimation and feedback for {FDD} massive {MIMO},''
  \emph{IEEE Trans. Signal Process.}, vol.~63, no.~23, pp. 6169--6183, Dec.
  2015.

\bibitem{Cheng2013}
P.~Cheng, Z.~Chen, Y.~Rui, Y.~Guo, L.~Gui, M.~Tao, and Q.~Zhang, ``Channel
  estimation for {OFDM} systems over doubly selective channels: A distributed
  compressive sensing based approach,'' \emph{IEEE Trans. Commun.}, vol.~61,
  no.~10, pp. 4173--4185, Oct. 2013.

\bibitem{Taubock2010a}
G.~Taubock, F.~Hlawatsch, D.~Eiwen, and H.~Rauhut, ``Compressive estimation of
  doubly selective channels in multicarrier systems: Leakage effects and
  sparsity-enhancing processing,'' \emph{IEEE J. Sel. Topics Signal Process.},
  vol.~4, no.~2, pp. 255--271, Apr. 2010.

\bibitem{Masood2015}
M.~Masood, L.~Afify, and T.~Al-Naffouri, ``Efficient coordinated recovery of
  sparse channels in massive {MIMO},'' \emph{IEEE Trans. Singal Process.},
  vol.~63, no.~1, pp. 104--118, Jan. 2015.

\bibitem{Rao2014}
X.~Rao and V.~Lau, ``Distributed compressive {CSIT} estimation and feedback for
  {FDD} multi-user massive {MIMO} systems,'' \emph{IEEE Trans. Signal
  Process.}, vol.~62, no.~12, pp. 3261--3271, June 2014.

\bibitem{Qi2014}
C.~Qi and L.~Wu, ``Uplink channel estimation for massive {MIMO} systems
  exploring joint channel sparsity,'' \emph{Electron. Lett.}, vol.~50, no.~23,
  pp. 1770--1772, Nov. 2014.

\bibitem{Gao2016}
Z.~Gao, L.~Dai, W.~Dai, B.~Shim, and Z.~Wang, ``Structured compressive
  sensing-based spatio-temporal joint channel estimation for {FDD} massive
  {MIMO},'' \emph{IEEE Trans. Commun.}, vol.~64, no.~2, pp. 601--617, Feb.
  2016.

\bibitem{Nan2015}
Y.~Nan, L.~Zhang, and X.~Sun, ``Efficient downlink channel estimation scheme
  based on block-structured compressive sensing for {TDD} massive {MU-{MIMO}}
  systems,'' \emph{IEEE Wireless Commun. Lett.}, vol.~4, no.~4, pp. 345--348,
  Aug. 2015.

\bibitem{Hou2014}
W.~Hou and C.~W. Lim, ``Structured compressive channel estimation for
  large-scale {MISO}-{OFDM} systems,'' \emph{IEEE Commun. Lett.}, vol.~18,
  no.~5, pp. 765--768, May 2014.

\bibitem{Qi2015}
C.~Qi, G.~Yue, L.~Wu, Y.~Huang, and A.~Nallanathan, ``Pilot design schemes for
  sparse channel estimation in {OFDM} systems,'' \emph{IEEE Trans. Veh.
  Technol.}, vol.~64, no.~4, pp. 1493--1505, Apr. 2015.

\bibitem{He2015}
X.~He, R.~Song, and W.~P. Zhu, ``Pilot allocation for distributed compressed
  sensing based sparse channel estimation in {MIMO-OFDM} systems,'' \emph{IEEE
  Trans. Veh. Technol.}, vol.~PP, no.~99, pp. 1--1, June 2015.

\bibitem{Duarte2011}
M.~Duarte and Y.~Eldar, ``Structured compressed sensing: From theory to
  applications,'' \emph{IEEE Trans. Signal Process.}, vol.~59, no.~9, pp.
  4053--4085, Sept. 2011.

\bibitem{Donoho2006}
D.~L. Donoho, M.~Elad, and V.~N. Temlyakov, ``Stable recovery of sparse
  overcomplete representations in the presence of noise,'' \emph{IEEE Trans.
  Inf. Theory}, vol.~52, no.~1, pp. 6--18, Jan. 2006.

\bibitem{Chen1998}
S.~S. Chen, D.~L. Donoho, and M.~A. Saunders, ``Atomic decomposition by basis
  pursuit,'' \emph{SIAM J. Sci. Comput.}, vol.~20, no.~1, pp. 33--61, Aug.
  1998.

\bibitem{Mallat1993}
S.~Mallat and Z.~Zhang, ``Matching pursuits with time-frequency dictionaries,''
  \emph{IEEE Trans. Signal Process.}, vol.~41, no.~12, pp. 3397--3415, Dec.
  1993.

\bibitem{Baron2005a}
D.~Baron, M.~B. Wakin, M.~F. Duarte, S.~Sarvotham, and R.~G. Baraniuk,
  ``Distributed compressed sensing,'' \emph{Preprint}, vol.~22, no.~10, pp.
  2729 -- 2732, 2005.

\bibitem{Gao2014}
Z.~Gao, L.~Dai, and Z.~Wang, ``Structured compressive sensing based
  superimposed pilot design in downlink large-scale {MIMO} systems,''
  \emph{Electron. Lett.}, vol.~50, no.~12, pp. 896--898, June 2014.

\bibitem{Qin2016}
Q.~Qin, L.~Gui, B.~Gong, X.~Ren, and W.~Chen, ``Structured distributed
  compressive channel estimation over doubly selective channels,'' \emph{IEEE
  Trans. Broadcast.}, vol.~62, no.~3, pp. 521--531, Sept 2016.

\bibitem{Dai2013}
L.~Dai, Z.~Wang, and Z.~Yang, ``Spectrally efficient time-frequency training
  ofdm for mobile large-scale mimo systems,'' \emph{IEEE J. Sel. Areas
  Commun.}, vol.~31, no.~2, pp. 251--263, Feb. 2013.

\bibitem{Marques2010}
Guidelines for Evaluation of Radio Transmission Technologies for {IMT}-2000,
  Recommendation {ITU-R} 1997.

\end{thebibliography}

\end{document}